\documentclass[letter, 11pt]{article}
\usepackage{amsmath} %equation*, pmatrix, cases
\usepackage{amssymb}  % \mathbb
\usepackage{amsthm}
\usepackage{graphicx}
\usepackage{times}
\usepackage{setspace} 
\usepackage{pifont}
\usepackage{multirow}
\usepackage{listings}
\newcommand{\cmark}{\ding{51}}%
\newcommand{\xmark}{\ding{55}}%

\newcommand{\bhline}[1]{\noalign{\hrule height #1}}

\newcommand{\defeq}{\stackrel{\mbox{\scriptsize{\normalfont\rmfamily def. }}}{=}}

\newcommand{\Order}{\mathrm{O}}
\newcommand{\Pro}{\mathbf{Pr}}
\newcommand{\E}{\mathbf{E}}

\newcommand{\Disc}{\mathrm{Disc}}
\newcommand{\difm}{\Delta}

\newcommand{\trM}{P} %rv of token flows

\newcommand{\VS}{V} %vertex set
\newcommand{\trE}{E^P}
\newcommand{\trN}{N^P} %N of P
\newcommand{\trd}{d^P} %degree of P
\newcommand{\trML}{P_{\mathrm{L}}} %rv of token flows
\newcommand{\trMM}{P_{\mathrm{M}}} %rv of token flows
\newcommand{\sE}{\lambda_2} %rv of token flows
\newcommand{\conf}{\mathcal{X}} %configuration of tokens
\newcommand{\cconf}{X} %configuration of tokens
\newcommand{\Y}{\mathcal{Y}} %configuration of tokens
\newcommand{\rvD}{\mathcal{D}} %rv of token flow
\newcommand{\rvr}{r} %rv of token flow
\newcommand{\rvDs}{\mathrm{D}} %rv of token flow
\newcommand{\IP}{\widetilde{P}} %rv of token flows
\newcommand{\ft}{\tau} %rv of token flows
\newcommand{\fv}{\nu} %rv of token flows
\newcommand{\fk}{\kappa} %rv of token flows

\def\df<#1>{\mathbf{1}\{#1\}}
\def\defeq{\mathrel{\mathop:}=}
\newcommand{\Not}{K} %# of tokens
\newcommand{\tT}{T}%rv of token flows
\newcommand{\Nov}{N}%rv of token flows

\newcommand{\Tc}{T} %rv of token flows

\newtheorem{theorem}{Theorem}[section]
\newtheorem{lemma}[theorem]{Lemma}
\newtheorem{corollary}[theorem]{Corollary}
\newtheorem{proposition}[theorem]{Proposition}
\newtheorem{observation}[theorem]{Observation}

\newtheorem{definition}[theorem]{Definition}

\title{Discrepancy Analysis of a New Randomized Diffusion Algorithm}
\author{
 Takeharu Shiraga\footnote{
    Department of Information and System Engineering, Faculty of Science and Engineering, 
   Chuo University, Tokyo, Japan\protect \\ 
  {\ttfamily shiraga@ise.chuo-u.ac.jp}} }\if0 \and 
 Yukiko Yamauchi\footnotemark[1] \and 
 Shuji Kijima\footnotemark[1] \and
 Masafumi Yamashita\footnotemark[1]
}\fi
\begin{document}
\makeatletter
\makeatother
\maketitle
%%%%%%%%%%
%%%%%%%%%%
\begin{abstract}
For an arbitrary initial configuration of discrete loads over vertices of a distributed graph, we consider the problem of minimizing the {\em discrepancy} between the maximum and minimum loads among all vertices. 
For this problem, this paper is concerned with the ability of natural diffusion-based iterative algorithms: at each discrete and synchronous time step on an algorithm, each vertex is allowed to distribute its loads to each neighbor (including itself) without occurring negative loads or using the information of previous time steps.

In this setting, this paper presents a new {\em randomized} diffusion algorithm like multiple random walks. 
Our algorithm archives $\Order(\sqrt{d \log \Nov})$ discrepancy for any $d$-regular graph with $\Nov$ vertices with high probability, while {\em deterministic} diffusion algorithms have $\Omega(d)$ lower bound.
Furthermore, we succeed in generalizing our algorithm to any symmetric round matrix.
This yields that $\Order(\sqrt{ d_{\max} \log \Nov})$ discrepancy for arbitrary graphs without using the information of maximum degree $d_{\max}$. 

\ 

\end{abstract}

\lstset{
	language=C,
	escapechar=\@,
	tabsize=2,
	backgroundcolor={\color[rgb]{1,1,1}},
	frame=tb,
	keywords={if, else, for}
}

\section{Introduction}
	This paper is concerned with the load balancing problem on distributed networks.
	Let $G=(\VS,E)$ be a connected graph with $\Nov=|\VS|$ vertices, and let $\conf^{(0)}_v$ denote a initial amount of loads on each $v\in V$.
	Then, from an arbitrary initial configuration of loads $\conf^{(0)}$, we consider iterative algorithms 
	which update the configuration of loads from $\conf^{(t)}$ to $\conf^{(t+1)}$ on each time step $t$
	with a goal to minimize the {\em discrepancy} between the maximum and minimum loads among all vertices as well as possible.
	%where $\conf^{(t)}$ denotes the configuration of loads at time step $t$ in the algorithms.
	Especially, this paper focuses on the {\em diffusion} algorithms: 
	in an update, each vertex distributes its loads to each neighbor synchronously.
	Because of not only its locality and simplicity, 
	but also deep connection with the theory of multiple random walks and mixing time of Markov chains, diffusion algorithms have been well studied recently.
	%%G=(V,E), \Nov, \conf is determined.
	%	
	\subsection{Previous works}
		%
		%Now, we introduce previous works and related works on the diffusion algorithms. 
		%Since most of previous works were studied on regular graphs, we firstly focus on the results on the regular graphs.
		\noindent \textbf{Continuous loads: }
		If the loads are {\em continuous} ($\conf^{(t)}_v\in \mathbb{R}$), i.e. divisible,  
		techniques to estimate the discrepancy corresponding to the theory of Markov chains have been well studied~(See e.g. \cite{SS94}).
		For example, on $d$-regular graphs, a natural diffusion algorithm 
		such that each vertex $v$ sends $\conf^{(t)}_v/(d+1)$ loads to each neighbor and keeps the same amount of loads for itself
		achieves a constant discrepancy after appropriate time steps. 
		Strictly speaking, for any vector $\xi\in \mathbb{R}^\Nov$, let
		\begin{eqnarray}
		\label{def:disc}
		\Disc(\xi)&\defeq&\max_{v,u\in V}|\xi_v-\xi_u|,
		\end{eqnarray}
		which represents the {\em discrepancy} between the maximum and minimum values of $\xi$ among all vertices.
		Then, 
		$\Disc(\conf^{(\Tc)})=\Order(1)$ after $\Tc\defeq \Order(\log(\Disc(\conf^{(0)})\Nov)/(1-\sE))$ steps, 
		where %$\Disc(\conf^{(t)})$ represents the discrepancy between the maximum and minimum loads among all vertices at time $t$
		%($\Disc(\conf^{(t)})\defeq \max_{v,u\in V}|\conf^{(T)}_v-\conf^{(T)}_u|$) and
		$\sE$ is the second largest eigenvalue of the transition matrix of the graph. The {\em convergence time} $\Tc$ is highly related to the {\em mixing time} of Markov chains (we will give precise discussions in Section~\ref{sec:continuous}).
		%These discussion is closely related to the mixing time of the random walk on the graph 	(We will give more precise discussion in Section~\ref{???}).

		\noindent \textbf{Discrete loads: }
		On the other hand, if the loads are {\em discrete} ($\conf^{(t)}_v\in \mathbb{Z}$), i.e. indivisible, 
		it is not easy to estimate the discrepancy compared with the continuous case although the analytic techniques are also related to the theory of Markov chains.
		Considering the discrete loads is important as practical settings and there are many previous works.

		%\begin{definition}[Discrete diffusion algorithm]
		%At each time step $t$, each vertex $v$ partitions $\conf^{(t)}_v\in \mathbb{Z}$ into 
		%$d$ integers $(f^{(t)}_{v,v_0}, f^{(t)}_{v,v_1},\ldots, f^{(t)}_{v,v_{d-1}})$, and sends $f^{(t)}_{v,v_i}$ loads to $v_i$ for each $i$. 
		%\end{definition}% s.t. $\sum_{i=0}^{d-1}f_{v,v_i}=\conf^{(t)}_v$ and $f_{v,v_i}\geq 0$, 

		For example, it is easy to consider a natural discretization diffusion algorithm such that 
		each vertex $v$ partitions $\conf^{(t)}_v\in \mathbb{Z}$ into $d+1$ integers 
		$\left \lceil \frac{\conf^{(t)}_v}{d+1}\right \rceil, \ldots, \left \lceil \frac{\conf^{(t)}_v}{d+1}\right \rceil, 
		\left \lfloor \frac{\conf^{(t)}_v}{d+1} \right \rfloor, \ldots, \left \lfloor \frac{\conf^{(t)}_v}{d+1} \right \rfloor$, 
		and sends each of them to each neighbor (including itself). 
		Rabani~et~al.~\cite{RSW98} gave a framework of the analysis to deal with the discrepancy of discrete load balancing algorithms 
		including this {\sc Send}$(\lceil \conf^{(t)}_v/(d+1) \rceil\ {\rm or}\ \lfloor \conf^{(t)}_v/(d+1) \rfloor)$ algorithm.
		From their result, it was shown that $\Disc(\conf^{(\Tc)})=\Order(d \log \Nov /(1-\sE))$ %($=\Order(\log(\Disc(\conf^{(0)})\Nov)/(1-\sE))$
		for any $d$ regular graph.
		Almost same but slightly refined upper bound on the discrepancy was given later by Shiraga et al.~\cite{SYKY13}. 
		They showed $\Disc(\conf^{(\Tc)})=\Order(d t_{\rm mix})$, 
		where $t_{{\rm mix}}$ is a mixing time of the graph. 
		It is well known that $t_{\rm mix}=\Order(\log \Nov/(1-\sE))$. 

		Berenbrink~et~al.~\cite{BCFFS15} studied a {\em randomized} diffusion algorithm to get smaller discrepancy with high probability, 
		i.e. with probability larger than $1-1/\Nov^c$ for some constant $c$. 
		In their algorithm, each vertex $v$ sends $\lfloor \conf^{(t)}_v/(d+1) \rfloor$ loads to each neighbor (including itself) firstly.
		Then the remaining loads are randomly sent one by one without replacing to neighbors (including itself).
		For this {\sc RSend}$(\lceil \conf^{(t)}_v/(d+1) \rceil\ {\rm or}\ \lfloor \conf^{(t)}_v/(d+1) \rfloor)$ algorithm on $d$ regular graphs, 
		they showed that $\Disc(\conf^{(T)})$ is bounded by $\Order\left(\left(d+\sqrt{\frac{d\log d}{1-\sE}}\right)\sqrt{\log \Nov}\right)$
		and $\Order((d\log \log \Nov)/(1-\sE))$ with high probability.
		Later, Sauerwald and Sun~\cite{SS14} gave an $\Order(d^2\sqrt{\log \Nov})$ discrepancy. % for this algorithm. 
		Note that this bound is independent to the expansion of graphs, i.e. independent of the second largest eigenvalue $\sE$.
		As an other randomized diffusion algorithm, Akbari and Berenbrink~\cite{AB13} dealt with a randomized ordering version of so called rotor-router model, 
		and showed that $\Disc(\conf^{(\Tc)})=\Order((d\log \log \Nov)/(1-\sE))$.
		%The upper bound of this model is $\Order(d\sqrt{log n})$, $\Order(d \log \log \Nov /(1-\lambda))$.
		
		A recent progress on diffusion algorithms was given by Berenbrink~et~al.~\cite{BKKMU15}. 
		They gave a strong framework of the deterministic diffusion algorithms and analyzed the discrepancy.
		%Their algorithm is classified as cumulatively fair and good s balancer.
		For example, for a proposed algorithm such that each vertex $v$ sends $\lfloor \conf^{(t)}_v/2d \rfloor$ loads to each neighbor and keeps the remaining loads for itself, 
		they showed that $\Disc(\conf^{(\Tc)})=\Order(d\sqrt{\log \Nov/(1-\sE)})$ on $d$ regular graphs.
		It means that this {\sc Send}$(\lfloor \conf^{(t)}_v/2d \rfloor)$ algorithm improves 
		the upper bound of the discrepancy of the {\sc Send}$(\lceil \conf^{(t)}_v/(d+1) \rceil\ {\rm or}\ \lfloor \conf^{(t)}_v/(d+1) \rfloor)$.
		In a sense, {\sc Send}$(\lfloor \conf^{(t)}_v/2d \rfloor)$ is a kind of {\em lazy} version of 
		{\sc Send}$(\lceil \conf^{(t)}_v/(d+1) \rceil\ {\rm or}\ \lfloor \conf^{(t)}_v/(d+1) \rfloor)$.
		%and the laziness plays a key role in their analysis. 
		Laziness is a famous property used in the field of random walks (lazy random walk stays current vertex with probability larger than $1/2$). 
		Furthermore, they gave an algorithm which achieves $\Order(d)$ the discrepancy within $\Order(\Tc+(\log \Nov)/(1-\sE))$ step.
		In this algorithm, each vertex $v$ sends $[\conf^{(t)}_v/3d]$ loads to each neighbor and keeps the remaining loads itself, where $[\cdot]$ denotes rounding to the nearest integer.
		They also showed that the lower bound of the discrepancy is $\Omega(d)$ for the deterministic (stateless) discrete diffusion algorithms, 
		hence the  {\sc Send}$([\conf^{(t)}_v/3d])$  gives a tight upper bound. % for the deterministic discrete diffusion.
		%They also gave a lazy version of the rotor-router model, i.e. each vertex has self loops more than $d/2$, 
		%and showed that showed that $\Order(d)$ discrepancy within $\Order\left(T+\frac{d\log^2\Nov}{1-\lambda}\right)$ step.
		%
		
		\noindent \textbf{On arbitrary graphs: }
		Since real computer networks often have the scale free property nowadays, the demand of studying load balancing algorithms on irregular graph is increasing. 
		%To study load balancing algorithms on irregular graphs has a demand since networks in the real world is irregular graph often having the scale free property.

		For the continuous case, it is not too difficult to discuss the discrepancy
		since the theory of {\em symmetric} (or {\em reversible}) Markov chains has been established in the framework including irregular graphs.
		Strictly speaking, let $P\in [0,1]^{V\times V}$ be a transition matrix (round matrix) on $V$. 
		Then, for the algorithm such that each vertex $v$ sends $\conf^{(t)}_vP_{v,u}$ loads to a neighbor $u$, 
		the discrepancy converges to a constant within $\Order(\Tc)$ step if $P$ is symmetric. %, i.e. $P_{v,u}=P_{u,v}$ holds for any $v,u\in V$. 
		For example, on an arbitrary graph, the algorithm such that each vertex $v$ sends $\conf^{(t)}_v/d_{\rm max}$ loads to each neighbor and keeps the remaining loads itself 
		achieves this property, where $d_{\rm max}$ is the maximum degree of the graph.
		If one would rather not use $d_{\max}$ since this is a global variable, 
		then the algorithm with $P_{v,u}=1/\min\{d_v,d_u\}$ for any $\{v,u\} \in E$ (and $P_{v,v}=1-\sum_{u:\{v,u\}\in E}P_{v,u}$)
		also converges to a constant discrepancy since $P$ is symmetric.
		This chain is called {\em Metropolis chain}. Note that this algorithm only require each vertex the knowledge of the degree of each neighbor (cf. \cite{NOSY10}). 

		On the other hand, for the discrete case, it is difficult to analyze on arbitrary graphs 
		since it is not clear that if one can generalize the analytic techniques for the discrete diffusion of regular graphs or not.
		%For the discrete case, there are not many studies concerned with the diffusion algorithm on arbitrary graphs because of its difficulty of the analysis.
		For the algorithm such that each vertex $v$ sends $\lceil \conf^{(t)}_vP_{v,u}\rceil$ or $\lfloor \conf^{(t)}_vP_{v,u} \rceil$ loads to a neighbor $u$, % for symmetric $P$, 
		the analysis of Rabani et. al~\cite{RSW98} showed that $\Disc(\conf^{(T)})=\Order\left(\Psi_1(P)\right)$, 
		where $\Psi_1(P)$ is called the {\em local-1 divergence} (See \eqref{def:local2} for the precise definition). 
		They also showed that $\Psi_1(P)=\Order(d_{\rm max}\log \Nov/(1-\sE))$ for any symmetric $P$.
		Shiraga et al.~\cite{SYKY13} showed that $\Disc(\conf^{(T)})=\Order\left(d_{\rm max} t_{\rm mix}\right)$ for this algorithm.
		%But it is not clear that if one can generalize the results on above discrete diffusion algorithms.
		%It is not clear that almost all above algorithms to like $\lceil \conf^{(t)}_v/d_{\max} \rceil$ algorithm.

	\subsection{Related works}
		\noindent \textbf{Algorithms with the state: }
		Note that all above algorithms are {\em stateless}, i.e. each vertex does not use any information of previous time steps, 
		while some previous works concerned with diffusion algorithms with the state.
		The {\sc Rotor-router model} is a typical one, which is a well studied deterministic process analogous to random walks. 
		%appears in several fields such as self stabilizing, graph exploration, etc.
		In this algorithm, each vertex sends loads one by one to neighboring vertices in the round robin fashion.
		%In other words, {\sc Rotor-router model} sends $\lfloor \conf^{(t)}_v/d \rfloor$ or $\lceil \conf^{(t)}_v/d \rceil$ loads to its neighbors, 
		{\sc Rotor-router model} has the same upper bound as {\sc Send}$(\lceil \conf^{(t)}_v/(d+1) \rceil\ {\rm or}\ \lfloor \conf^{(t)}_v/(d+1) \rfloor)$ (cf. \cite{RSW98, SYKY13}).
		Berenbrink~et~al.~\cite{BCFFS15} gave a lazy version of the rotor-router model {\sc LRotor-router model}, i.e. each vertex has self loops more than $d/2$, 
		and showed that $\Order(d)$ discrepancy within $\Order(\Tc+(d\log^2\Nov)/(1-\sE))$ steps.
		
		\noindent \textbf{Algorithms occurring negative loads: }
		We also note that all above algorithms satisfy the property such that each tokens distributes its own loads to each neighbor (and itself).
		There are several previous works corresponding to the algorithms occurring {\em negative loads}, 
		i.e. each vertex possibly sends {\em more than its own loads} to each neighbor (and itself). 
		% in order to obtain a refined upper bound of the discrepancy.
		For example, consider the algorithm such that each vertex sends randomly rounded
		($\lfloor \conf^{(t)}_v/(d+1) \rfloor$ or $\lceil \conf^{(t)}_v/(d+1) \rceil$) loads to each neighbor and itself.
		This algorithm possibly occurs negative loads, i.e. total amount of the sent loads possibly becomes $\lceil \conf^{(t)}_v/(d+1) \rceil (d+1)>\conf^{(t)}_v$.
		Sauerwald and Sun showed that $\Order(\sqrt{d \log \Nov})$ discrepancy within $\Order(\Tc)$ with high probability for this algorithm. 
		They also showed that on arbitrary graphs, the algorithm rounding $\lceil \conf^{(t)}/2d_{\max}\rceil $ or $\lfloor \conf^{(t)}/2d_{\max}\rfloor$ with appropriate probability
		archives $\Disc(\conf^{(T)})=\Order(\Psi_2(P)\sqrt{\log \Nov})$ with high probability, where $\Psi_2(P)$ denotes 
		the {\em local $2$-divergence} (See \eqref{def:local2} for the precise definition), 
		and showed that $\Psi_2(P)=\Order(d_{\rm max})$ for the round matrix $P_{v,u}=1/cd_{\max}$ for any $\{v,u\}\in E$ (and the remaining is the self loop), 
		where $c$ is some constant.
		%However, it is not clear that the upper bound of the $\Psi_2(P)$ for (general) symmetric $P$.
		%There are some other results concerned with the randomized diffusion algorithms which occurs negative loads \cite{FGS12, ABS13, SS14}.
		The algorithm in \cite{ABS13} corresponds to a diffusion algorithm with the state and negative loads. 
		This algorithm achieves $\Order(\sqrt{d_{\max} \log \Nov})$ discrepancy for arbitrary graphs.
		
		\noindent \textbf{Matching algorithms: }
		Matching based algorithms have been well studied as well as diffusion algorithms.
		Matching-based algorithms generate a matching of the graph in a distributed way at each round, 
		and the endpoints of each matching edge balances loads as evenly as possible. 
		Friedrich and Sauerwald~\cite{FS09} studied randomized version of matching models on regular graphs.
		They showed that $\Disc(\conf^{(T)})=\Order(\Psi_2(P)\sqrt{\log \Nov})$ with high probability. % where $\Psi_2(P)$ denotes the local $2$-divergence.
		They also showed that $\Psi_2(P)=\Order(\sqrt{d/(1-\sE)})$.
		The results of Sauerwald and Sun~\cite{SS14} is the best result of the discrepancy so far. 
		They showed that a constant discrepancy within $\Order(T)$ step for a randomized matching model on regular graphs.

	%%%%%%%%%%%%%%%%%%%%%
	\subsection{This work}
	%%%%%%%%%%%%%%%%%%%%%
		\noindent \textbf{Motivation: }
		The strength of diffusion algorithms is its strong locality.
		%The ability required to each vertex 
		The ability required to each vertex is only to count its own loads (degree) and to send its own loads to each neighbor.
		There is no need to communicate and check the amount of neighbor's loads like matching algorithms. 
		
		Our main concern of this paper is to investigate the discrepancy of the discrete diffusion algorithms with simplest assumptions, 
		i.e. stateless and non-negative loads diffusion algorithms. 
		We call these algorithms {\em natural} diffusion algorithms.
		This framework contains Markov chains (multiple random walks).
		
		In previous works, the $\Omega(d)$ lower bound of the discrepancy for any {\em deterministic} natural diffusion algorithms on $d$ regular graphs has been shown~\cite{BKKMU15}.
		Similarly, one can guess a $\Omega(d)$ bound for any {\em randomized} natural diffusion algorithm
		since the discrepancies of all previous upper bounds of randomized natural diffusion algorithms depend on the polynomial of $d$.
		However, a lower bound of randomized natural diffusion algorithms has not been known, 
		i.e. no one knows that whether there is a randomized natural diffusion algorithm with the discrepancy ${\rm o}(d)$ or not.
		
		\noindent \textbf{Results: }
		For this question, this paper proposes a {\em new} randomized natural diffusion algorithm
		which archives $\Order(\sqrt{d\log \Nov})$ discrepancy within $\Order(\Tc)$ steps with high probability on $d$ regular graphs.
		This result gives a positive answer to the question, i.e. breaks $\Omega(d)$ barrier,
		since the result guarantees that the discrepancy is ${\rm o}(d)$ for $d=\omega(\log N)$ regular graphs.
		Surprisingly, even though we compared with the best upper bound of the diffusion algorithm allowing negative loads~\cite{SS14} or with the state~\cite{ABS13}, 
		this upper bound for our natural diffusion algorithm is the same magnitude. 
		
		Furthermore, we succeed in generalizing the proposed algorithm for arbitrary symmetric round matrices. % and giving an upper bound of the discrepancy.
		This allows us to construct a randomized natural diffusion algorithms on arbitrary graphs with 
		$\Order(\sqrt{d_{\rm max}\log \Nov})$ discrepancy within $\Order(\Tc)$ steps with high probability. % which is the best upper bound on the discrepancy.
		\subsubsection{Result on regular graphs}
			First, we introduce a new randomized natural diffusion algorithm (Algorithm~1).
			The main idea of this algorithm is to add the {\em laziness} to {\sc RSend}$(\lceil \conf^{(t)}_v/(d+1) \rceil\ {\rm or}\ \lfloor \conf^{(t)}_v/(d+1) \rfloor)$.
			%
			%\noindent \textbf{Definition of Algorithm~1: }
			Let $G=(V,E)$ is an arbitrary $d$-regular graph with $|\VS|=\Nov$ vertices.
			For each $v\in V$, let $v_0, v_1, \ldots, v_{d-1}$ denote the $d$ neighbors of $v$.
			
			\noindent \textbf{Definition of Algorithm~1: }
			Let $\conf^{(0)}\in \mathbb{Z}^\Nov_{\geq 0}$ be a initial configuration of $\Not$ loads over $\VS$, 
			and let $\conf^{(t)}\in \mathbb{Z}^\Nov_{\geq 0}$ denote the configuration of $\Not$ loads over $\VS$ at time $t\in \mathbb{Z}_{\geq 0}$ in our algorithm.
			In an update from $\conf^{(t)}$ to $\conf^{(t+1)}$ in our algorithm, at each vertex $v$, 
			each load $k$ ($k\in \{0,1,\ldots,\conf^{(t)}_v-1\}$) randomly samples a number $r^{(t)}_v(k)$ from the interval $\left[ \frac{k}{\conf^{(t)}_v},  \frac{k+1}{\conf^{(t)}_v} \right)$.
			Then, each load $k$ moves to a neighbor $v_i$ if sampled random number is in the interval of $v_i$, 
			i.e. each load $k$ moves to $v_i$ if $r_v(k)\in \left[ \frac{i}{2d},  \frac{i+1}{2d} \right)$ (if $r_v(k)\in [1/2,1)$, load $k$ stays at $v$).
%			
\if0
\begin{lstlisting}[caption={(on $d$ regular graphs $G=(V,E)$)}, label=Alg1]
At each time step @$t\geq 0$@ and at each vertex @$v\in V$@
	for @$k\in \{0,1,\ldots, \conf^{(t)}_v-1\}$@:
		load @$k$@ randomly samples a number @$r_v^{(t)}(k)$@ from @$[k/\conf^{(t)}_v, (k+1)/\conf^{(t)}_v)$@
		for @$i\in \{0,1,\ldots, d-1\}$@:
			if @$r_v^{(t)}(k) \in [i/2d, (i+1)/2d)$@: load @$k$@ moves to vertex @$v_i$@
				/* @if $r_v^{(t)}(k) \in [1/2,1)$, then load $k$ stays $v$@ */
\end{lstlisting}
\fi
%

			Obviously, Algorithm~1 is a natural diffusion algorithm. 
			If each load randomly samples a number from $[0,1)$, then this is multiple {\em lazy} random walks.
			Laziness is a important property to estimate the discrepancy (cf.~\cite{BKKMU15}). %, and plays a key role in our analysis.
			For this algorithm, we showed the following upper bound of the discrepancy.
			\begin{theorem}[Result on regular graphs]
				\label{thm:reg}
				Suppose that $G=(V,E)$ is an arbitrary connected $d$-regular graph.
				Then, for any $\conf^{(0)}$ and for each $T\geq \frac{\log(4\Disc(\conf^{(0)}) \Nov)}{1-\sE}$, 
				$\conf^{(T)}$ of Algorithm~1 satisfies that 
				\begin{eqnarray*}
					\Pro\left[ \Disc(\conf^{(T)}) \leq 18\sqrt{d\log \Nov} \right]\geq 1-\frac{2}{\Nov}.
				\end{eqnarray*}
			\end{theorem}
%NOTENOTENOTENOTENOTENOTENOTENOTENOTENOTENOTENOTENOTENOTENOTENOTENOTENOTE
			%From Theorem~\ref{thm:reg} gives the following meanings for diffusion algorithms on $d$ regular graphs. % for the framework of the natural diffusion algorithms
			%First, although all deterministic natural diffusion algorithms can not break $\Omega(d)$, randomization allows us break the $\Omega(d)$,
			%i.e. $\Omega(d)$ does not hold for randomized natural diffusion algorithms in general. % (for $d=\omega(\log \Nov)$)
			%Second, $\Order(\sqrt{d\log \Nov})$ upper bound of our {\em natural} diffusion algorithm matches 
			%the best result on previous works allowing negative loads or the state of vertices.
%NOTENOTENOTENOTENOTENOTENOTENOTENOTENOTENOTENOTENOTENOTENOTENOTENOTENOTE
			%Note that Algorithm~\ref{Alg 1} is stateless and non-negative loads.
			%Surprisingly, even though we compared with the best upper bound of the diffusion algorithm allowing negative loads~\cite{SS14}, 
			%the upper bound of Theorem~\ref{thm:reg} for Algorithm~1 satisfies the same magnitude. 
			%Surprisingly, we can achieve $\Order(\sqrt{d \log \Nov})$ discrepancy without occurring negative loads nor without history.
			Similar to natural deterministic diffusion algorithms, using laziness improves the previous upper bounds of randomized natural diffusion algorithms \cite{BCFFS15, AB13,SS14}.
			Table~\ref{table:1} summarizes the discrepancies of previous results and this work on regular graphs.
%\if0
			\begin{table}[t]
				%{\renewcommand\arraystretch{1}
				\centering
				\begin{tabular}{lclccl}%325
				\bhline{1.5pt}
				Algorithm&
				D/R&
				$\Disc(\conf^{(\Tc)})$&
				NL&
				SL&
				Ref. 
				\\ \bhline{1.2pt}
				%%%%%%%
				{\sc Send}$\left(\left \lceil \frac{\conf^{(t)}_v}{d+1} \right \rceil\ {\rm or}\ \left \lfloor \frac{\conf^{(t)}_v}{d+1} \right \rfloor\right)$
				&D
				&$\Order\left(d\frac{\log \Nov}{1-\sE}\right)$
				&\cmark
				&\cmark
				&\cite{RSW98}
				\\ 
				%%%%%%%%
				{\sc Send}$\left(\left \lfloor \frac{\conf^{(t)}_v}{2d} \right \rfloor \right)$
				&D
				&$\Order\left(d\sqrt{\frac{\log \Nov}{1-\sE}}\right)$
				&\cmark
				&\cmark
				&\cite{BKKMU15}
				\\ 
				%%%%%%%%%%
				{\sc Send}$\left(\left[ \frac{\conf^{(t)}_v}{3d} \right] \right)$
				 &D
				 &$\Order(d)\ $\footnotemark[1]
				&
				 \cmark&
				 \cmark&
				\cite{BKKMU15}
				\\ 
				%%%%%%%%%%
				{\sc RSend}$\left(\left \lceil \frac{\conf^{(t)}_v}{d+1} \right \rceil\ {\rm or}\ \left \lfloor \frac{\conf^{(t)}_v}{d+1} \right \rfloor\right)$
				&R
				&$\Order\left(\left(d+\sqrt{\frac{d\log d}{1-\sE}}\right)\sqrt{\log \Nov}\right)$
				&\cmark
				&\cmark
				&\cite{BCFFS15}
				\\ 
				%%%%%%%%%%55
				
				&
				&$\Order\left(d\frac{\log \log \Nov}{1-\sE}\right)$
				&
				&
				&\cite{BCFFS15}
				\\
				%%%%%%%%%%%%%%
				
				&
				&$\Order\left(d^2\sqrt{\log \Nov}\right)$
				&
				&
				&\cite{SS14}
				\\
				%%%%%%%%%%%%
				{\sc RRotor-router}
				&R
				&$\Order\left(d\frac{\log \log \Nov}{1-\sE}\right)$
				&\cmark
				&\cmark
				&\cite{AB13}
				\\ \bhline{1.2pt}
				%%%%%%%%%%%%%
				Algorithm 1
				&R
				&$\Order \left(\sqrt{d\log \Nov}\right)$
				&\cmark
				&\cmark
				&Theorem~\ref{thm:reg}
				\\ \bhline{1.2pt}
				%%%%%%%%%%%%%
				{\sc Rotor-router}
				&D
				&$\Order\left(d\frac{\log \Nov}{1-\sE}\right)$
				&\cmark
				&\xmark
				&\cite{RSW98}
				\\ 
				%%%%%%%%%%%%%
				{\sc LRotor-router}
				&D
				&$\Order(d)\ $\footnotemark[2]
				&\cmark
				&\xmark
				&\cite{BKKMU15}
				\\ 
				%%%%%%%%%%%%%
				{\sc Algorithm} in \cite{ABS13}
 				&R
				&$\Order \left(\sqrt{d\log \Nov}\right)$
				&\xmark
				&\xmark
				&\cite{ABS13}
				\\
				%%%%%%%%%%%%%
				{\sc Neg.RSend}
				&R
				&$\Order \left(\sqrt{d\log \Nov}\right)$
				&\xmark
				&\cmark
				&\cite{SS14}
				\\ 
				\bhline{1.5pt}
				\end{tabular}
%}
				\caption{Previous works and our result on $d$-regular graphs. D/R represents a deterministic/randomized algorithm. NL represents an algorithm which does not occur negative loads. SL represents state less algorithms.}
				\label{table:1}
			\end{table}		
%\fi
		\subsubsection{Generalized algorithm and the result on arbitrary graph}\label{sec:genresult}
			Furthermore, we succeeded in generalizing Algorithm~1 for any {\em round matrix} $P$ (Algorithm~2).
			
			\noindent \textbf{Notations: }
			Let $\VS$ be a vertex set, and let $\Nov=|\VS|$.
			Let $\trM \in \mathbb{R}_{\geq 0}^{\Nov \times \Nov}$ be a round (transition) matrix on $V$, 
			i.e. $\sum_{u\in \VS}\trM_{v,u}=1$ holds for any $v\in \VS$, where $\trM_{v,u}$ denotes $(v,u)$ entry of $\trM$.
			%In this paper we assume that $\trM^0$ is the identity matrix.
			%In the sequence of the studies on the load balancing, $P$ is also called the {\em round matrix}.
			For a round matrix $P$ on $V$, let $\trN_v$ be the set of neighbors of $v\in V$, i.e. $\trN_v\defeq \{u\in V\mid \trM_{v,u}>0\}$. 
			In this paper, we assume an arbitrary ordering on $\trN_v$, i.e. we denote $\trN_v=\{v_0, v_1,\ldots, v_{\trd_v-1}\}$, where $\trd_v=|\trN_v|$.
			%Our generalized algorithm is defined of a vertex set $\VS$ and a round (transition) matrix $\trM$ over $\VS$, i.e. 
			%Let $\trN_v$ be the set of neighbors of $v\in V$, i.e. $\trN_v\defeq \{u\in V\mid \trM_{v,u}>0\}$. 
			%In this paper, we assume an arbitrary ordering on $\trN_v$, i.e. let we denote $\trN_v=\{v_0, v_1,\ldots, v_{\trd_v-1}\}$, where $\trd_v=|\trN_v|$.
			%Now, let $\conf^{(t)}_v$ denote the number of tokens on $v$ at time $t$.
\footnotetext[1]{Discrepancy after $\Order(\Tc+(\log \Nov)/(1-\sE))$ steps.}
\footnotetext[2]{Discrepancy after $\Order(\Tc+(d\log^2\Nov)/(1-\sE))$  steps.}	

			\noindent \textbf{Definition of Algorithm~2: }
			Let $\conf^{(0)}\in \mathbb{Z}^\Nov_{\geq 0}$ be a initial configuration of $\Not$ loads over $\VS$, 
			and let $\conf^{(t)}\in \mathbb{Z}^\Nov_{\geq 0}$ denote the configuration of $\Not$ loads over $\VS$ at time $t\in \mathbb{Z}_{\geq 0}$ in our algorithm.
			In an update from $\conf^{(t)}$ to $\conf^{(t+1)}$ in our algorithm, at each vertex $v$, 
			each load $k$ ($k\in \{0,1,\ldots, \conf^{(t)}_v-1\}$) randomly samples a random number $r_v(k)$ from the interval $\left[ \frac{k}{\conf^{(t)}_v},  \frac{k+1}{\conf^{(t)}_v} \right)$.
			Then, each load $k$ moves to its corresponding neighbor $v_i$, i.e. load $k$ moves to $v_i$ if $r_v^{(t)}(k)\in \left[\sum_{j=0}^{i-1}P_{v,v_j}, \sum_{j=0}^{i}P_{v,v_j} \right)$ 
			(let $\sum_{j=0}^{-1}P_{v,v_j}=0$).
%
\if0
\begin{lstlisting}[caption={(according to the round matrix $P$)}, label=Alg2]
At each time step @$t$@ and at each vertex @$v$@
	for @$k\in \{0,1,\ldots, \conf^{(t)}_v-1\}$@:
		load @$k$@ randomly samples a number @$r_v^{(t)}(k)$@ from @$[k/\conf^{(t)}_v, (k+1)/\conf^{(t)}_v)$@
		for @$i\in \{0,1,\ldots, d^P_v-1\}$@:
			if @$r_v^{(t)}(k) \in \left[\sum_{j=0}^{i-1}P_{v,v_j}, \sum_{j=0}^{i}P_{v,v_j} \right)$@: load @$k$@ moves to vertex @$v_i$@
\end{lstlisting}
\fi
%
%
\begin{figure}[t]
\centering
\includegraphics[width=16.0cm]{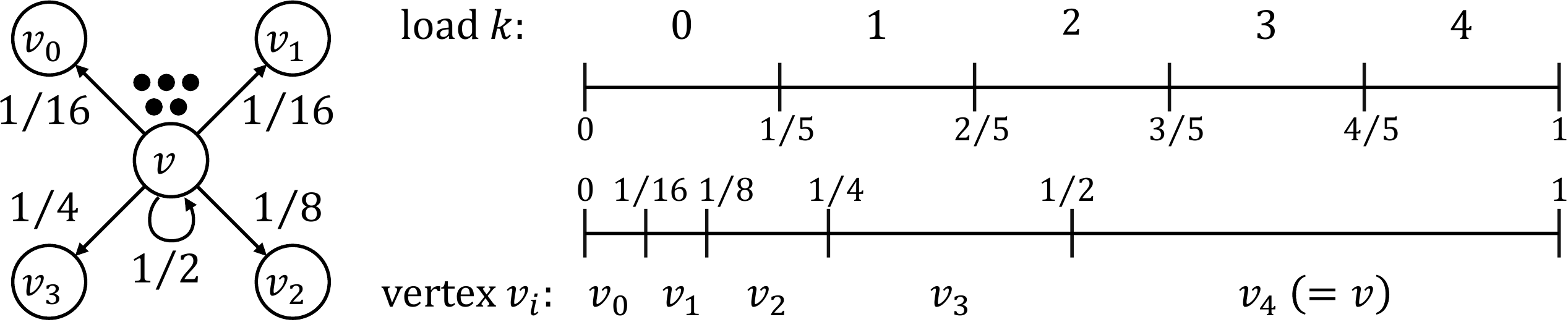}
\caption{In this example, there are $5$ loads at $v$. Then, load $0$ moves $v_0$ / $v_1$ / $v_2$ w.p. $\frac{5}{80}\cdotp 5$ / $\frac{5}{80}\cdotp 5$ / $\frac{6}{80}\cdotp 5$. Load $1$ moves $v_2$ / $v_3$ w.p. $\frac{1}{20}\cdotp 5$ / $\frac{3}{20}\cdotp 5$. Load $2$ moves to $v_3$ / $v_4$ w.p. $\frac{1}{2}$ / $\frac{1}{2}$. Load $3$ / $4$ moves to $v_4$ w.p. $1$.}
\label{fig:Alg2}
\end{figure}

			This is a generalization of Algorithm~1. We put an example on Figure~\ref{fig:Alg2}.
			Only different point compared with multiple random walks according to $P$ is that each load $k$ randomly samples a number from $[k/\conf^{(t)}_v, (k+1)/\conf^{(t)}_v)$, 
			while each load randomly samples a number from $[0,1)$ in multiple random walks. %, then this algorithm is the same as multiple random walks according to $P$. 
			%See Figure 1 for an example. There are $5$ tokens ($k=0,1,2,3,4$) on $v$ and $(P_{v,v_0}, P_{v,v_1}, P_{v,v_2}, P_{v,v_3})=(1/10, 3/5, 1/5, 1/10)$.
			%In this example, token $0$ moves to $v_0$ ($v_1$) with probability $1/2$, token $1$ ($2$) moves to $v_1$ with probability $1$, 
			%token $3$ moves to $v_1$ ($v_2$) with probability 		$1/2$, 
			%and token $4$ moves to $v_2$ ($v_3$) with probability $1/2$.
			%More precisely, we assume that for each $v\in \VS$, each $u$ is mapped into a interval in $[0,1)$ whose length is $P_{v,u}$. 
			%Then each token $k$ moves to $u$ if the generated random number in $[k/\conf^{(t)}_v, (k+1)/\conf^{(t)}_v)$ is in the interval of $u$.
			%(See equation~(\ref{def:rvD}) in Section~\ref{subsec:modeldescription} for the concrete definition).
			%This algorithm is quite simple.
			%At each round, each vertex only sends its tokens using the number of tokens on it and the probability distribution around it,  
			%does not need to communicate to neighbors or use the number of tokens of its neighbor.
			%Furthermore, 
			%Almost all previous works are corresponding to ``uniform'' transition matrices, i.e., 
			%$\trM_{v,x}=\trM_{v,y}$ holds for any $x,y$ in the neighboring set of $v$ ($x,y\neq v$).

			One of the main reason introducing the generalized algorithm is to use the (lazy) {\em Metropolis chain} $\trMM$ on arbitrary graphs.
			This chain is defined by as follows: 
			$(\trMM)_{v,u}\defeq 1/(2\max\{d_v,d_u\})$ for any $\{v,u\}\in E$, $(\trMM)_{v,v}\defeq 1-\sum_{u:\{v,u\}\in E}(\trMM)_{v,u}$ for any $v\in V$ 
			and $(\trMM)_{v,u}\defeq 0$ for any $\{v,u\}\notin E$.
			The main strength of this chain is that there is no need to require each vertex to know the maximum degree over all vertices.
			Then, we show the following upper bound of Algorithm~2 according to $\trMM$ on arbitrary graphs.
			\begin{theorem}[Results on arbitrary graphs]
				\label{thm:gen}
				Suppose that $G=(V,E)$ is an arbitrary connected graph. % and the transition matrix is $\trMM$.
				Then, for any $\conf^{(0)}$ and for each $\Tc\geq \frac{\log(4\Disc(\conf^{(0)}) \Nov)}{1-\sE}$, 
				$\conf^{(T)}$ of Algorithm~2 according to $\trMM$ satisfies that 
				\begin{eqnarray*}
					\Pro\left[ \Disc(\conf^{(T)})  \leq 16\sqrt{d_{\max}\log \Nov} \right]\geq 1-\frac{2}{\Nov}.
				\end{eqnarray*}
			\end{theorem}
			This upper bound dramatically improves the previous works on the framework of the natural diffusion adaptive to the metropolis chain: 
			%since the upper bound of the discrepancy of the only natural diffusion algorithm in previous works adaptive to the metropolis chain 
			$\Order(d_{\rm max}\log \Nov/(1-\sE))$ for {\sc DSend}$(\lceil \conf^{(t)}_v(\trMM)_{v,u}\rceil$ or $\lfloor \conf^{(t)}_v(\trMM)_{v,u} \rceil)$.
			Even though we compared with the best upper bound of the diffusion algorithm allowing negative loads and the information of the maximum degree~\cite{SS14}, 
			the upper bound of Theorem~\ref{thm:gen} for Algorithm~2 satisfies the same magnitude. 
			%compared with the diffusion algorithm allowing negative loads and the information of the maximum degree~\cite{SS14}. 
			%This also matches the previous result of the diffusion model on general graph~\cite{SS14} without negative loads.
			%Furthermore, each vertex only needs the degree of its neighbors to calculate $(\trMM)_{x,y}=(1/2)\max\{d_x,d_y\}$,  
			%while previous work~\cite{SS14} needs $d_{\max}$, which is the maximum value of the degree among all vertices.
			%
		%
		\subsubsection{Idea of the proof and technical contribution}
			Above our main theorems are shown by a load (token)-based analysis.
			The load configuration of Algorithm~2 is determined by random variables of the {\em destinations} of each load at each vertex (Observations~\ref{obs:defconf}).
			The properties of the destinations of each load are described in 
			Observations~\ref{obs:exD} (Expectation of the destinations) and ~\ref{obs:indepD} (Independency of the destinations). 
			Especially, Observations~\ref{obs:indepD} plays a key role to prove the key lemma (Lemma~\ref{lemm:mainM}) stating the Martingale difference.
			This and the well-known concentration inequality (Azuma–Hoeffding inequality) allow us to get the following good upper bound of 
			the Discrepancy between discrete and continuous diffusions in general form.
			%The proof of Lemma~\ref{lemm:mainM} and Theorem~\ref{thm:main}.
			\begin{theorem}[Discrepancy between discrete and continuous diffusions]
				\label{thm:main}
				For any $\conf^{(0)}$, for any round matrix $\trM$ and for each time $T\in \mathbb{Z}_{\geq 0}$, 
				$\conf^{(T)}$ of Algorithm~2 according to $P$ satisfies that 
				\begin{eqnarray*}
					\Pro\left[ \max_{w\in V}\Bigl|\conf^{(\tT)}_w-(\conf^{(0)}\trM^\tT)_w\Bigr| \leq 4\Psi_2(\trM)\sqrt{\log \Nov} \right] \geq 1-\frac{2}{\Nov}.
				\end{eqnarray*}
			\end{theorem}
			$\conf^{(0)}\trM^\tT$ is the configuration of loads on the continuous diffusion according to $P$ (See Section~\ref{sec:continuous} for the detail). 
			Note that this value also represents the expected configuration of multiple random walks.

			As an technical contribution, we obtain a general upper bound of the $\Psi_2(P)$.
			The definition of the local $p$-divergence is as follows.
			\begin{definition}[local $p$-divergence, \cite{RSW98, FS09}]\label{def:local2}
				For any $p\in \mathbb{Z}_{\geq 0}$, the local-$p$ divergence of $P$ is defined by
				\begin{eqnarray*}
					\Psi_p(P)\defeq \max_{w\in V}\left( \sum_{t=0}^\infty \sum_{(v,u)\in V\times V: \trM_{v,u}>0}|\trM^t_{v,w}-\trM^t_{u,w}|^p \right)^{1/p}.
				\end{eqnarray*}
			\end{definition}
			Although $\Psi_p(P)=\Order(\sqrt{cd_{\max}})$ for the transition matrix on $G=(V,E)$
			such that $P_{v,u}=1/cd_{\max}$ for any $\{v,u\}\in E$ and $P_{v,v}$ is the remaining probability for any $v\in V$~\cite{SS14}, 
			%there exists a tight bound for the local-2 divergence of the uniform matrices~\cite{SS14}, 
			%i.e. $P_{v,v_i}=P_{v,v_j}$ holds for any $v_i, v_j \in \trN_v$ except for $v$ 
			%itself\footnote{e.g. $P_{v,u}=1/d_{\max}$ for any $u\in \trN_v$ and $P_{v,v}$ is the remaining probability.},
			this proof is not simple and it was not clear that if we can extend this proof to any round matrix like Metropolis chain. 

			For this problem, we realized that the equation corresponding to the local-2 divergence can be transformed into a simple equation 
			by the idea of the {\em Dirichlet form} of the {\em reversible} transition matrix (Lemma~\ref{lemm:dirichlet}). 
			Then, adding the monotonicity of the lazy transition matrix ($P^{t+1}_{v,v}\geq P^{t}_{v,v}$), 
			we succeeded in extending the previous work of the upper bound of the local-2 divergence as follows.
			\begin{theorem}[Upper bound of the local 2 divergence]
				\label{thm:local2}
				Suppose that $\trM$ is reversible and lazy, and let $\pi$ be the stationary distribution of $P$.
				%Let $P_{\min}\defeq \min_{(v,u)\in \VS \times \VS:P_{v,u}\neq 0}P_{v,u}$.
				Then, it holds that
				\begin{eqnarray*}
					\Psi_2(P)
					&\leq &\sqrt{ \frac{2\max_{w\in V}\pi_{w}}{\min_{(v,u)\in \VS\times \VS:P_{v,u}>0}\pi_v\trM_{v,u}}}.
				\end{eqnarray*}
			\end{theorem}
			%We are surprised that this quite simple proof based on fundamental techniques of Markov chains. 
			%This generalizes the previous upper bound of $\Psi_2(P)$ to any reversible and lazy chain.
			Note that if $\trM$ is symmetric, Theorem~\ref{thm:local2} becomes simpler since $\pi_v=1/\Nov$ for any $v\in \VS$. 
			\begin{corollary}\label{cor:local2}
			Suppose that $\trM$ is symmetric and lazy.
			%Let $P_{\min}\defeq \min_{(v,u)\in \VS \times \VS:P_{v,u}\neq 0}P_{v,u}$.
			Then, it holds that
			\begin{eqnarray*}
			\Psi_2(P)	&\leq & \sqrt{\frac{2}{\min_{(v,u)\in \VS\times \VS:P_{v,u}>0}\trM_{v,u}}}.
			\label{thm:local2sym}
			\end{eqnarray*}
			\end{corollary}
			Corollary~\ref{cor:local2} allows us to estimate an upper bound of the (lazy) Metropolis chain $\trMM$, 
			$\Psi_2(\trMM)=\Order (\sqrt{d_{\max}})$.
			We believe that our simple transformation of the local-2 divergence (Lemma~\ref{lemm:dirichlet}) has applications in other problems.
			%Combining this bound of the local-2 divergence and Theorem~\ref{thm:main} gives our upper bound on arbitrary graphs (Theorem~\ref{thm:gen}). 
			
\if0
			\begin{table}[ht]
				\begin{tabular}{lllllll}
					\bhline{1.5pt}
					Algorithm&$t$&Disc. after $t$&$\Psi_p(P_{\rm M})$&NL&D/R&Ref. \\ \hline
					{\sc DSend}$(\lceil \conf^{(t)}_vP_{v,u}\rceil$ or $\lfloor \conf^{(t)}_vP_{v,u} \rceil)$& $\Order(T)$&$\Order\left(\Psi_1(P)\right)$&$\Order(d_{\max}\frac{\log \Nov }{1-\lambda})$&\cmark&D&\cite{RSW98}\\ \hline
					%SS14& $\Order \left(\Psi_2(P)\sqrt{\log \Nov }\right)$&$\Order(\sqrt{d_{\max}})$& \xmark& \xmark& R& \cite{SS14}\\ \hline \hline
					Algorithm 2& $\Order(T)$&$\Order \left(\Psi_2(P)\sqrt{\log \Nov}\right)$& $\Order(\sqrt{d_{\max}})$&\cmark& R& Thm.~\ref{thm:main}\\ \bhline{1.5pt} 
				\end{tabular}
				\caption{Previous works of natural diffusion-based algorithms in general case}
			\end{table}
\fi
%Note%Note%Note%Note%Note%Note%Note%Note%Note%Note%Note
		\if0\noindent \textbf{Organization of this paper: }
		This paper is organized as follows.
		In Section~\ref{sec:preliminaries}, we give a precise discussion of the continuous diffusion and some basic and important properties of our algorithms to prove the main theorem.
		In Section~\ref{sec:mainthm}, we proves Theorem~\ref{thm:main}.
		Section~\ref{sec:local2} is concerned with the upper bound of the local-2 divergence.
		In Section~\ref{sec:final}, we conclude the proof of Theorem~\ref{thm:reg} and Theorem~\ref{thm:gen}.\fi
%Note%Note%Note%Note%Note%Note%Note%Note%Note%Note%Note%Note%Note

\section{Continuous diffusion and Markov chains}\label{sec:continuous}
	%
	%\subsection{Continuous diffusion and Markov chains}\label{sec:continuous}
	%
		Before describing our algorithms, we give a precise discussion of the {\em continuous} diffusion algorithm according to the round matrix $P$. %this is highly related to our analysis.
		Let $\cconf^{(0)}\in \mathbb{R}^n_{\geq 0}$ be a initial configuration of $\Not$ $(=\sum_{v\in V}\cconf^{(0)}_v)$ loads over a vertex set $V$, 
		and let $\cconf^{(t)}\in \mathbb{R}^n_{\geq 0}$ denotes the configuration of the loads at time $t$ on the diffusion algorithm according to the round matrix $P$.
		%Let $\Not$ be the total amount of loads, i.e.  $\Not \defeq \sum_{v\in V}\conf^{(0)}_v$.
		Then, for each $t$ and $v\in V$, $\cconf^{(t+1)}_v$ is defined by $\cconf^{(t)}$ and the round matrix $P$, 
		$
			(\cconf^{(t+1)})_v \defeq (\cconf^{(t)}\trM)_v\ =\ \sum_{u\in V}\cconf^{(t)}_u\trM_{u,v}.
		$
		From this definition, the configuration of loads at time $t$ is described by the initial configuration of loads and the round matrix as follows.
		\begin{eqnarray}
			\cconf^{(t)}&=&\cconf^{(0)}\trM^t.
		\end{eqnarray}
		Our main concern is {\em the discrepancy} $\Disc(\cconf^{(t)})$ defined in \eqref{def:disc}.
		The limit behavior of the discrepancy is characterized by the rich theory of the convergence of the Markov chains.
		To discuss clearly, we introduce some terminologies. 
		A $P$ is called {\em irreducible} if for any $v,u\in V$ there exists a $t$ such that $P^t_{v,u}>0$. 
		Note that $P$ is irreducible if and only if the transition diagram ${\cal G}=(V, \trE)$ is connected, where $\trE\defeq \{(v,u)\in V\times V\mid P_{v,u}>0\}$.
		A $P$ is called {\em symmetric} if $P_{v,u}=P_{u,v}$ holds for any $v,u\in V$.
		Let $\lambda_1, \lambda_2,\ldots, \lambda_\Nov$ be the eigenvalues of $P$. We assume $|\lambda_1|\geq |\lambda_2|\geq \cdots \geq |\lambda_\Nov|$.
		Then, it is easy to derive the following proposition stating the convergence time of the diffusion algorithm according to $P$.
		We put the proof in Appendix~\ref{sec:continuous2}.
		\begin{proposition}[The discrepancy of the continuous diffusion algorithm \cite{SS94, RSW98}]
		\label{prop:cont}
			Suppose that $P$ is irreducible and symmetric.
			Then, for any $\cconf^{(0)}\in \mathbb{R}^n_{\geq 0}$, $w\in V$, $\varepsilon\in (0,1)$ and 
			$
				T \geq \frac{1}{1-\sE}\log\left(\frac{4\Disc(\cconf^{(0)}) \Nov}{\varepsilon}\right), 
			$
			$
				\Disc(\cconf^{(T)})= \Disc(\cconf^{(0)} P^T)\leq \varepsilon
			$ holds.
		\end{proposition}
		%We check the following proposition in the full version.
		Note that combining Proposition~\ref{prop:cont}, Theorem~\ref{thm:main} (Discrepancy between discrete and continuous diffusions) 
		and Theorem~\ref{thm:local2} (Upper bound of the local 2 divergence), 
		it is easy to obtain our main Theorems~\ref{thm:reg} (Result on regular graphs) and ~\ref{thm:gen} (Results on arbitrary graphs). 
		%We put this proof in Section~\ref{sec:final}.
		Thus, we start proving Theorem~\ref{thm:main} (Section~\ref{sec:mainthm}) and Theorem~\ref{thm:local2} (Section~\ref{sec:local2}).
\section{Proof of Theorem~\ref{thm:main}}\label{sec:mainthm}
	\subsection{Key properties of our model}
		In this section, we observe key properties of our model (Algorithm~2).
		The properties are described by Observations~\ref{obs:defconf},~\ref{obs:exD} and \ref{obs:indepD}.
		Each of them is concerned with the random variable $\rvD^{(t)}_v(k)$, 
		which denotes the {\em destination} of $(k+1)$-th token on $v\in V$ at $t\geq 0$. %,  playing a key role to prove Theorem~\ref{thm:main}.
		%Let $\conf^{(0)}\in \mathbb{Z}^\Nov_{\geq 0}$ be a initial configuration of $\Not$ tokens over $\VS$, 
		%and let $\conf^{(t)}\in \mathbb{Z}^\Nov_{\geq 0}$ denote the configuration of $\Not$ tokens over $\VS$ at time $t\in \mathbb{Z}_{\geq 0}$ in our algorithm.
		%In an update from $\conf^{(t)}$ to $\conf^{(t+1)}$ of our algorithm, 
		%each load $k$ ($k\in \{0,1,\ldots, \conf^{(t)}_v-1\}$) randomly samples a random number $r_v(k)$ from the interval $	\left[ \frac{k}{\conf^{(t)}_v},  \frac{k+1}{\conf^{(t)}_v} \right)$.
		For the notational convenience, we define $\IP_{v,u}$, which denotes the interval of $\trM_{v,u}$ on $[0,1)$.
		For any $v\in V$ and $u=v_i\in \trN_v$, let 
		\begin{eqnarray}
		\label{def:Pint}
			\IP_{v,u}&\defeq &\left[\sum_{j=0}^{i-1}P_{v,v_j}, \sum_{j=0}^{i}P_{v,v_j}\right).
		\end{eqnarray}
		We assume $\sum_{j=0}^{-1}P_{v,v_j}\defeq 0$. Note that the length of $\IP_{v,u}$ is equal to $P_{v,u}$.

		Recall that $r^{(t)}_v(k)$ is a random number sampled from $[k/\conf^{(t)}_v, (k+1)/\conf^{(t)}_v)$, 
		where $\conf^{(t)}$ denotes the configuration of loads at time $t$ in Algorithm~2 (see the definition of Algorithm~2 in Section~\ref{sec:genresult}).
		Then, $\rvD^{(t)}_v(k)$ is defined as follows. 
		%Note that $\conf^{(t)}$ denotes the configuration of loads at time $t$ in Algorithm~2 
		%and $r^{(t)}_v(k)$ is a random number sampled from $[k/\conf^{(t)}_v, (k+1)/\conf^{(t)}_v)$ (definition of Algorithm~2 in Section~\ref{sec:genresult}).
		\begin{eqnarray}
			\label{def:rvD}
			\rvD^{(t)}_v(k)=u\ \ \mathrm{if}\ \ \rvr_v^{(t)}(k)\in \IP_{v,u}.
			%\rvD^{(t)}_v(k)&=&u\ \mathrm{if}\ \rvr_v^{(t)}(k)\in \left[ \sum_{j=0}^{i-1}\trM_{v,u}, \sum_{j=0}^{i}\trM_{v,u} \right).
		\end{eqnarray}
		%In our analysis, $\rvD^{(t)}_v(k)$ plays a key role. 
		%From the definition, 
		%Thus
		%\begin{eqnarray}
		%	\sum_{u\in V}\sum_{k=0}^{\conf^{(t)}_v-1}\df<\rvD^{(t)}_v(k)=u>
		%	&=&\sum_{k=0}^{\conf^{(t)}_v-1}\sum_{u\in \trN_v}\df<\rvD^{(t)}_v(k)=u>
		%	\ =\ \conf^{(t)}_v.
		%	\label{def:partition}
		%\end{eqnarray}
		Now, let we observe the following $3$ properties of $\rvD^{(t)}_v(k)$.
		First one is concerned with the connection between the configuration of loads and $\rvD^{(t)}_v(k)$.
		Note that $\sum_{k=0}^{\conf^{(t)}_v-1}\df<\rvD^{(t)}_v(k)=u>$ denotes the number of loads sent from $v$ to $u$ at time $t$ in Algorithm~2.
		\begin{observation}[Relation between the configuration and destinations]
			For each time step $t\geq 0$, $\conf^{(t+1)}$ of Algorithm~2 is determined by $\conf^{(t)}$ and 
			$\left \{\rvD^{(t)}_v(0), \rvD^{(t)}_v(1), \ldots, \rvD^{(t)}_v(\conf^{(t)}_v-1)\right \}_{v\in V}$, i.e. for each time $t\geq 0$ and $u\in V$, 
			\begin{eqnarray}
				\conf^{(t+1)}_u&= &\sum_{v\in V}\sum_{k=0}^{\conf^{(t)}_v-1}\df<\rvD^{(t)}_v(k)=u>.
				%\label{def:conf}
			\end{eqnarray}
			\label{obs:defconf}
		\end{observation}
		Next one describe the expected value of $\rvD^{(t)}_v(k)$ conditioned on $\conf^{(t)}_v$.
		\begin{observation}[Expectation of the destination]
			For any $t\geq 0$, $v\in V$, $k\in \{0,1,\ldots \conf^{(t)}_v-1\}$ and $u\in \trN_v$, it holds that
			\begin{eqnarray*}
				\E\left[\df<\rvD^{(t)}_v(k)=u>\mid \conf^{(t)}_v\right]
				&=&\left|\left[\frac{k}{\conf^{(t)}_v}, \frac{k+1}{\conf^{(t)}_v}\right) \cap \IP_{v,u} \right|\cdot \conf^{(t)}_v.
			\end{eqnarray*}
			\label{obs:exD}
		\end{observation}
		\begin{proof}
			\begin{eqnarray*}
				\E\left[\df<\rvD^{(t)}_v(k)=u>\mid \conf^{(t)}_v\right]
				&=&\Pro\left[\rvD^{(t)}_v(k)=u\mid \conf^{(t)}_v\right]
				\ =\ \Pro\left[\rvr^{(t)}_v(k) \in \IP_{v,u} \mid \conf^{(t)}_v\right]
			\end{eqnarray*}
			holds and we obtain the claim since $\rvr^{(t)}_v(k)$ is randomly sampled from $\left[ \frac{k}{\conf^{(t)}_v},  \frac{k+1}{\conf^{(t)}_v} \right)$.
		\end{proof}
		The last one shows that the conditional independency of each destination in Algorithm~2.
		\begin{observation}[Independency of the destinations]
			Suppose that $\conf^{(t)}$ is fixed. 
			Then, for any $v,v'\in V$, $k\in \{0,1,\ldots \conf^{(t)}_v-1\}$ and $k' \in \{0,1,\ldots \conf^{(t)}_{v'}-1\}$, 
			$\rvD^{(t)}_v(k)$ and $\rvD^{(t)}_{v'}(k')$ are independent if $v\neq v'$ or $k\neq k'$.
			\label{obs:indepD}
		\end{observation}
		\begin{proof}
		For any intervals $[a,b)\subseteq [0,1)$, $[a',b')\subseteq [0,1)$ and $v, v', k, k'$ s.t. $v\neq v'$ or $k\neq k'$, 
		\begin{eqnarray}
		\lefteqn{\Pro\left[\left(r^{(t)}_v(k)\in [a,b)\right)\ {\rm and}\ \left(r^{(t)}_{v'}(k')\in [a',b')\right)\mid \conf^{(t)}\right]} \nonumber \\
		&=&\Pro\left[r^{(t)}_v(k)\in [a,b)\mid \conf^{(t)}\right]\cdotp \Pro\left[r^{(t)}_{v'}(k')\in [a',b')\mid \conf^{(t)}\right]
		\end{eqnarray}
		since $r^{(t)}_v(k)$ and $r^{(t)}_{v'}(k')$ are randomly sampled from each fixed interval independently. %Thus we obtain the claim.
		%If $v\neq v'$, it holds that $\Pro[\rvD^{(t)}_v(k)=v_i \cap \rvD^{(t)}_{v'}(k)={v'}_j\mid \conf^{(t)}]=\Pro[\rvD^{(t)}_v(k)=v_i\mid \conf^{(t)}]\Pro[\rvD^{(t)}_{v'}(k)={v'}_j\mid \conf^{(t)}]$.
		%If $k\neq k'$, it holds that $\Pro[\rvD^{(t)}_v(k)=v_i \cap \rvD^{(t)}_{v}(k')={v}_j\mid \conf^{(t)}]=\Pro[\rvD^{(t)}_v(k)=v_i\mid \conf^{(t)}]\Pro[\rvD^{(t)}_{v}(k')={v}_j\mid \conf^{(t)}]$.
		\end{proof}
		The analysis of this paper is derived from the properties described in Observations~\ref{obs:defconf}-\ref{obs:indepD}.
		%\subsection{ on the discrepancy analysis}
		%\label{subsec:modeldescription}
		%\noindent \textbf{Key properties: }
		%This paper is concerned with the behavior of $\conf^{(t)}$ defined by $\rvD^{(t)}_v(k)$ satisfying Observation~\ref{obs:indepD} and Proposition~\ref{obs:exD}.  
		%\noindent \textbf{Key properties for the discrepancy analysis: }
		%The main concern of this paper is the discrepancy of the Algorithm~2 $\Disc(\conf^{(t)})$.
		%To estimate $\Disc(\conf^{(t)})$, we introduce the following key lemma, 
		%which connects the configuration of the discrete diffusion and the configuration of the continuous diffusion.
		%From the consequence of Observation~\ref{obs:exD}-\ref{obs:defconf}, we obtain the following key properties for the proof of Theorem~\ref{thm:main}. 
		For example, the expected value of the configuration of loads is derived from Observations~\ref{obs:defconf} and Observations~\ref{obs:exD}.
		\begin{lemma}[Expectation of the configuration]
			\label{lemm:exconf}
			For any $T\geq 0$ and initial configuration $\conf^{(0)}$, it holds that
			\begin{eqnarray*}
				\E\left[\conf^{(T)}\right]&=&\conf^{(0)}P^T.
			\end{eqnarray*}
		\end{lemma}
		\begin{proof}
			Since
			\begin{eqnarray*}
				\E\left[\conf^{(t+1)}_u\mid \conf^{(t)}\right]
				&=& \sum_{v\in V}\sum_{k=0}^{\conf^{(t)}_v-1}\E\left[\df<\rvD^{(t)}_v(k)=u>\mid \conf^{(t)}\right] \\
				&=& \sum_{v\in V}\sum_{k=0}^{\conf^{(t)}_v-1}\left|\left[\frac{k}{\conf^{(t)}_v}, \frac{k+1}{\conf^{(t)}_v}\right) \cap \IP_{v,u} \right|\cdot \conf^{(t)}_v
				\ =\ \sum_{v\in V}\conf^{(t)}_vP_{v,u}
				\ =\ (\conf^{(t)}P)_u
			\end{eqnarray*}
			holds from Observations~\ref{obs:defconf} and~\ref{obs:exD}, we have
			\begin{eqnarray}
				\E\left[\conf^{(t+1)}\right]
				\ =\ \E\left[\E\left[\conf^{(t+1)}\mid \conf^{(t)}\right]\right]
				\ =\ \E\left[\conf^{(t)}\right]P
				\label{eq:aaawww}
			\end{eqnarray}
			and we obtain the claim by iterating the equation \eqref{eq:aaawww}.
		\end{proof}
		At the end of this section, we introduce the following $2$ lemmas, each of which guarantees an upper bound of the discrepancy between $\df<\rvD^{(t)}_v(k)=u>$ and its expectation.
		\begin{lemma}[Discrepancy around one load]
			\label{lemm:discload}
			For any $t\geq 0$, $v\in V$ and $k\in \{0,1,\ldots \conf^{(t)}_v-1\}$, it holds that
			\begin{eqnarray*}
				\sum_{u\in \trN_v}\left| \df<\rvD^{(t)}_v(k)=u>-\E\left[\df<\rvD^{(t)}_v(k)=u>\mid \conf^{(t)}_v\right] \right|&\leq & 2.
			\end{eqnarray*}
		\end{lemma}
		\begin{proof}	
		Since 
		\begin{eqnarray*}
		&&\sum_{u\in \trN_v}\df<\rvD^{(t)}_v(k)=u>\ =\ 1\ \ \ {\rm and} \\
		&&\sum_{u\in \trN_v}\E\left[\df<\rvD^{(t)}_v(k)=u>\mid \conf^{(t)}_v\right]\ =\ \sum_{u\in \trN_v}\Pro\left[\rvD^{(t)}_v(k)=u\right]\ =\ 1
		\end{eqnarray*}
		hold from the definition \eqref{def:rvD}, we obtain the claim.
		Note that 
		\begin{eqnarray*}
			\left| \df<\rvD^{(t)}_v(k)=u>-\E\left[\df<\rvD^{(t)}_v(k)=u>\mid \conf^{(t)}_v\right] \right|
			&\leq & \left| \df<\rvD^{(t)}_v(k)=u>\right|+\left| \E\left[\df<\rvD^{(t)}_v(k)=u>\mid \conf^{(t)}_v\right] \right|.
		\end{eqnarray*}
		\end{proof}
		\begin{lemma}[Discrepancy around one neighbor]
			\label{lemm:discneigh}
			For any $t\geq 0$, $v\in V$ and~$u\in \trN_v$, it holds that
			\begin{eqnarray*}
				\sum_{k=0}^{\conf^{(t)}_v-1}\left| \df<\rvD^{(t)}_v(k)=u>-\E\left[\df<\rvD^{(t)}_v(k)=u>\mid \conf^{(t)}_v\right] \right|&\leq & 2.
			\end{eqnarray*}
		\end{lemma}
		\begin{proof} 
		Let $\IP_{v,u}=[a_u,b_u)$, and let
		\begin{eqnarray*}
		A^{(t)}_v(u)&\defeq &\left \{k\in \{0,\ldots, \conf^{(t)}_v-1\}\mid \frac{k}{\conf^{(t)}_v}<a_u<\frac{k+1}{\conf^{(t)}_v} \right \}\ \ \ {\rm and} \\
		B^{(t)}_v(u)&\defeq &\left \{k\in \{0,\ldots, \conf^{(t)}_v-1\}\mid \frac{k}{\conf^{(t)}_v}<b_u<\frac{k+1}{\conf^{(t)}_v} \right \}.
		\end{eqnarray*}
		Note that both $|A^{(t)}_v(u)|\leq 1$ and $|B^{(t)}_v(u)|\leq 1$ hold 
		since $\left[\frac{k}{\conf^{(t)}_v},\frac{k+1}{\conf^{(t)}_v}\right)$ and $\left[\frac{k'}{\conf^{(t)}_v},\frac{k'+1}{\conf^{(t)}_v}\right)$ are disjoint for any $k\neq k'$.

		Now, we observe that if $k\notin \bigl(A^{(t)}_v(u)\cup B^{(t)}_v(u)\bigr)$, then {\bf 1.} $\left[\frac{k}{\conf^{(t)}_v},\frac{k+1}{\conf^{(t)}_v}\right)\subseteq \IP_{v,u}$ 
		or {\bf 2.} $\left[\frac{k}{\conf^{(t)}_v},\frac{k+1}{\conf^{(t)}_v}\right) \cap \IP_{v,u}=0$.
		For the case {\bf 1.}, 
		$\Pro[\rvD^{(t)}_v(k)=u\mid \conf^{(t)}_v]=1$ from the Observation~\ref{obs:exD}. Then, $\df<\rvD^{(t)}_v(k)=u>=1$ since $\rvD^{(t)}_v(k)$ must be $u$.
		For the case {\bf 2.}, 
		$\Pro[\rvD^{(t)}_v(k)=u\mid \conf^{(t)}_v]=0$ from the Proposition~\ref{obs:exD}. Then, $\df<\rvD^{(t)}_v(k)=u>=0$ since $\rvD^{(t)}_v(k)$ must not be $u$.
		Hence $\df<\rvD^{(t)}_v(k)=u>-\E[\df<\rvD^{(t)}_v(k)=u>\mid \conf^{(t)}_v]=0$ for both cases. Thus
		\begin{eqnarray*}
			\lefteqn{\sum_{k=0}^{\conf^{(t)}_v-1}\left| \df<\rvD^{(t)}_v(k)=u>-\E\left[\df<\rvD^{(t)}_v(k)=u>\mid \conf^{(t)}_v\right] \right|}\nonumber \\
			&=&\sum_{\substack{k\in \\ A^{(t)}_v(u)\cup B^{(t)}_v(u)}}\left| \df<\rvD^{(t)}_v(k)=u>-\E\left[\df<\rvD^{(t)}_v(k)=u>\mid \conf^{(t)}_v\right] \right| %\nonumber \\
			\leq |A^{(t)}_v(u)\cup B^{(t)}_v(u)|
			\leq 2.
		\end{eqnarray*}
		\end{proof}
	%This section gives the proof of Theorem~\ref{thm:main}. 
	%
	\subsection{Framework of the proof}
		Now, we estimate the discrepancy between $\conf^{(T)}_w$ and $(\conf^{(0)}\trM^T)_w$ for the Theorem~\ref{thm:main}.
		For the convenience, we introduce an useful notation.
		Let $V=\{0,1,\ldots,\Nov - 1\}$.
		Then, for any $\ft\in \{0,1,\ldots, \tT-1\}$, $\fv\in \VS$, and $\fk\in \{0,1,\ldots, \Not-1\}$, let
		\begin{eqnarray}
			\label{def:rvDs}
			\rvDs_{\Not \Nov \ft+\Not \fv +\fk}\ \defeq\ 
			\begin{cases}
				\rvD^{(\ft)}_\fv(\fk) & (\mathrm{if}\ \fk\in \{0,1,\ldots, \Not-1\}) \\
				-1 & (\mathrm{otherwise})
			\end{cases}.
		\end{eqnarray}
		This definition means that $\df<\rvDs_{\Not \Nov \ft+\Not \fv +\fk}=u>=0$ for any $\conf^{(t)}_\fv\leq \fk \leq \Not-1$ and $u\in V$.
		Then, the following is easily observed from Observation~\ref{obs:defconf} and the definition \eqref{def:rvDs}.
		\begin{observation}
			For any $t\in \mathbb{Z}_{\geq 0}$, $\conf^{(t)}$ is determined by 
			$\conf^{(0)}$ and $\rvDs_{0}, \rvDs_{1}, \ldots, \rvDs_{\Not \Nov t -1}$.
			\label{obs:confM}
		\end{observation}
		Note that $\rvDs_{\Not \Nov t -1}=\rvDs_{\Not \Nov (t -1)+\Not(\Nov-1)+(\Not-1)}$, which is $\rvD^{(t-1)}_{\Nov-1}(\Not-1)$ or $-1$.

		The main idea to estimate $\conf^{(T)}_w-(\conf^{(0)}\trM^T)_w$ is applying the Azuma-Hoeffding inequality (See Appendix~\ref{sec:concentration})
		to the martingale $\Y_0, \Y_1, \ldots, \Y_{\Not \Nov T}$ respect to $\rvDs_0, \rvDs_1, \ldots, \rvD_{\Not \Nov T-1}$, where
		\begin{eqnarray}
			\Y_\ell&\defeq &\E\left[\conf^{(T)}_w-(\conf^{(0)}P^T)_w\mid \rvDs_{0}, \rvDs_{1}, \ldots, \rvDs_{\ell-1}\right].
		\end{eqnarray}
%%%%%%%NOTENOTENOTENOTENOTENOTEAZUMANOTENOTENOTENOTENOTENOTENOTENOTENOTENOTENOTENOTE
		%\begin{theorem}[Asuma-Hoeffding Inequality, \cite{MU17}]
		%	Let $X_0, \ldots, X_n$ be a martingale such that
		%	$
		%	|X_k-X_{k-1}|\leq c_k.
		%	$
		%	Then, for all $t\geq 1$ and any $\lambda>0$, 
		%	$
		%		\Pro\left[|X_t-X_0|\geq \lambda \right]\leq 2\exp\left[-\frac{\lambda^2}{2\sum_{k=1}^t(c_k)^2}\right].
		%	$
		%\end{theorem}
%%%%%%%NOTENOTENOTENOTENOTENOTEAZUMANOTENOTENOTENOTENOTENOTENOTENOTENOTENOTENOTENOTE
		We assume that $\Y_{0}\defeq \E[\conf^{(T)}_w-(\conf^{(0)}P^T)_w]$.
		Since $\Y_{\Not \Nov \tT}=\conf^{(\tT)}_w-(\conf^{(0)}\trM^\tT)_w$ from Observation~\ref{obs:confM}
		and $\Y_{0}= \E[\conf^{(T)}_w-(\conf^{(0)}P^T)_w]=0$ from Lemma~\ref{lemm:exconf}, % (Expectation of the configuration), 
		\begin{eqnarray}
			\Pro\bigl[ |\Y_{\Not \Nov \tT}-\Y_{0} | \geq \eta \bigr] \leq 2\exp\left[ -\eta^2 / 2\sum_{\ell=0}^{ \Not \Nov \tT} (c_\ell)^2 \right]
			\label{eq:ah1111}
		\end{eqnarray}
		for any $\eta>0$ from Azuma-Hoeffding inequality, where $c_\ell$ is a value satisfies $\left| \Y_{\ell+1} - \Y_{\ell} \right|\leq c_\ell$. Hence
		%Since $\Y_{\Not \Nov \tT}=\conf^{(\tT)}_w-(\conf^{(0)}\trM^\tT)_w$ from Observation~\ref{obs:confM}
		%and $\Y_{0}=\E[\conf^{(\tT)}_w-(\conf^{(0)}\trM^\tT)_w]=0$, 
		%(See Lemma~\ref{lemm:frame} and (\ref{eq:dczero}) in Section~\ref{subsec:prof34} for the detail), 
		\begin{eqnarray}
			\Pro\left[ \max_{w\in V}|\conf^{(\tT)}_w-(\conf^{(0)}\trM^\tT)_w| \geq 2\sqrt{ \sum_{\ell=0}^{ \Not \Nov \tT-1} (c_\ell)^2 \log \Nov  } \right] \leq \frac{2}{\Nov}
			\label{eq:ah1}
		\end{eqnarray}
		by taking $\eta=\sqrt{ 2\sum_{\ell=0}^{ \Not \Nov \tT -1} (c_\ell)^2 \log \Nov ^2 }$ and using the union bound.
		%
		%Thus, using the union bound to (\ref{eq:ah1}) and apply the following lemma, we obtain Theorem~\ref{thm:main}.
		%\begin{lemma}
		%	\label{lemm:main}
		%	For any $T\in \mathbb{Z}_{\geq 0}$ and $w\in V$, it holds that
		%	\begin{eqnarray*}
		%		\sum_{\ell=0}^{ \Not \Nov\tT -1} \bigl( \Y_{\ell+1}(w,T) -\Y_\ell(w,T) \bigr)^2
		%		&\leq & 4\bigl(\Psi_2(P)\bigr)^2.
		%	\end{eqnarray*}
		%\end{lemma}
		%

		Thus, our main concern to obtain Theorem~\ref{thm:main} is the upper bound of the difference 
		$\Y_{\ell+1} - \Y_{\ell}=\E[\conf^{(T)}_w-(\conf^{(0)}P^T)_w\mid \rvDs_0, \ldots, \rvDs_\ell ]-\E[\conf^{(T)}_w-(\conf^{(0)}P^T)_w\mid \rvDs_0, \rvDs_1, \ldots, \rvDs_{\ell-1}]$.
		For this key value, we showed the following Lemma.
		\begin{lemma}[Martingale difference]
		\label{lemm:mainM}
			For any $\ft\in \{0,1,\ldots, T-1\}$, $\fv\in \{0,1,\ldots, \Nov - 1\}$ and $\fk\in \{0,1,\ldots, \Not-1\}$, 
			let $\ell=\Not \Nov \ft+\Not \fv+\fk$. Then, it holds that
			\begin{eqnarray*}
				\lefteqn{ \E\left[\conf^{(T)}_w-(\conf^{(0)}P^T)_w\mid \rvDs_0, \rvDs_1, \ldots, \rvDs_\ell \right]
				-\E\left[\conf^{(T)}_w-(\conf^{(0)}P^T)_w\mid \rvDs_0, \rvDs_1, \ldots, \rvDs_{\ell-1} \right] }\\
				&=&\sum_{u\in \trN_\fv}\left( \df<\rvDs_{\ell}=u>-\E\left[\df<\rvDs_\ell=u>\mid \conf^{(\ft)}_\fv\right] \right)\left(\trM^{T-\ft-1}_{u,w}-\trM^{T-\ft-1}_{\fv,w}\right).
			\end{eqnarray*}
		\end{lemma}
		Combining Lemma~\ref{lemm:mainM}, \eqref{eq:ah1}, Lemma~\ref{lemm:discload} and Lemma~\ref{lemm:discneigh}, we can show Theorem~\ref{thm:main}.
		\begin{proof}[Proof of Theoem~\ref{thm:main}]
			From the Lemma~\ref{lemm:mainM}, Cauchy-schwarz inequality and Lemma~\ref{lemm:discneigh}, % (Discrepancy around one token), 
			%it holds that
			\begin{eqnarray}
			\left( \Y_{\ell+1} -\Y_{\ell}\right)^2
			%&=&
			&\leq &\left( \sum_{u\in \trN_\fv}\left|\df<\rvDs_{\ell}=u>-\E\left[\df<\rvDs_\ell=u>\mid \conf^{(\ft)}_\fv\right]\right|
						\left|\trM^{T-\ft-1}_{u,w}-\trM^{T-\ft-1}_{\fv,w}\right|\right)^2\nonumber \\
			&\leq &\sum_{u\in \trN_\fv}\left|\df<\rvDs_{\ell}=u>-\E\left[\df<\rvDs_\ell=u>\mid \conf^{(\ft)}_\fv\right] \right| \cdotp \nonumber \\
			&&\sum_{u\in \trN_\fv}\left|\df<\rvDs_{\ell}=u>-\E\left[\df<\rvDs_\ell=u>\mid \conf^{(\ft)}_\fv\right] \right|\left|\trM^{T-\ft-1}_{u,w}-\trM^{T-\ft-1}_{\fv,w}\right|^2 \nonumber \\
			&\leq &2\sum_{u\in \trN_\fv}\left|\df<\rvDs_{\ell}=u>-\E\left[\df<\rvDs_\ell=u>\mid \conf^{(\ft)}_\fv\right] \right|\left(\trM^{T-\ft-1}_{u,w}-\trM^{T-\ft-1}_{\fv,w}\right)^2.
		\end{eqnarray}
		Thus, using Lemma~\ref{lemm:discneigh} and note the definition of $\rvDs_\ell$, 
		\begin{eqnarray}
			\lefteqn{ \sum_{\ell=0}^{\Not \Nov T-1}\left( \Y_{\ell+1} -\Y_{\ell}\right)^2 
			\ =\  \sum_{\ft=0}^{T-1}\sum_{\fv=0}^{\Nov - 1}\sum_{\fk=0}^{\Not -1}\left( \Y_{\Not \Nov \ft+\Not \fv+\fk+1} -\Y_{\Not \Nov \ft+\Not \fv+\fk}\right)^2} \nonumber \\
			&\leq &\sum_{\ft=0}^{T-1}\sum_{\fv=0}^{\Nov - 1}\sum_{\fk=0}^{\conf^{(\ft)}_\fv -1}%\nonumber \\
			\left(2\sum_{u\in \trN_\fv}\left|\df<\rvD^{(\ft)}_\fv(\fk)=u>-\E\left[\df<\rvD^{(\ft)}_\fv(\fk)=u>\mid \conf^{(\ft)}_\fv\right] \right|\left(\trM^{T-\ft-1}_{u,w}-\trM^{T-\ft-1}_{\fv,w}\right)^2 \right)\nonumber \\
			&\leq&4\sum_{\ft=0}^{T-1}\sum_{\fv=0}^{\Nov - 1}\sum_{u\in \trN_\fv}\left(\trM^{T-\ft-1}_{u,w}-\trM^{T-\ft-1}_{\fv,w}\right)^2
			\ \leq \ \bigl(2\Psi_2(P)\bigr)^2. \label{eq:comp}
		\end{eqnarray}
		Combining \eqref{eq:ah1} and \eqref{eq:comp}, we obtain the claim.
	\end{proof}
\if0
\begin{proof}[Proof of Theorem]
Since $\trM$ is symmetric from the assumption, we have
\begin{eqnarray*}
\lefteqn{\left|\conf^{(T)}_v-\conf^{(T)}_u\right|}\\
&=&\left|\left(\conf^{(T)}_v-(\conf^{(0)}P^T)_v\right)-\left((\conf^{(0)}P^T)_v-(\conf^{(0)}P^T)_u\right)-\left((\conf^{(0)}P^T)_u-\conf^{(T)}_u\right)\right|\\
&\leq &\left|\conf^{(T)}_v-(\conf^{(0)}P^T)_v\right|+\left|(\conf^{(0)}P^T)_v-(\conf^{(0)}P^T)_u\right|+\left|(\conf^{(0)}P^T)_u-\conf^{(T)}_u\right| \\
&\leq & 4\Psi_2(P)+\varepsilon.
\end{eqnarray*}
\end{proof}
\fi

	To complete the proof, we prove Lemma~\ref{lemm:mainM} in the following subsection.
	\subsection{Proof of Lemma~\ref{lemm:mainM}}
	\label{subsec:prof34}
		Lemma~\ref{lemm:mainM} is shown by carefully discussions of the conditional expectations characterized by the following Lemmas derived from Observations~\ref{obs:defconf}-\ref{obs:indepD}.
		First, we introduce the following lemma, 
		which characterizes the difference between $\conf^{(T)}_w$ and $(\conf^{(0)}P^T)_w$ 
		by the summation of the geometric series of the round matrix and $\df<\rvD^{(t)}_v(k)=u>-\E[\df<\rvD^{(t)}_v(k)=u>\mid \conf^{(t)}_v]$.
		\begin{lemma}[Relation between the discrete diffusion and the continuous diffusion]
			\begin{eqnarray*}
			\lefteqn{\conf^{(T)}_w-(\conf^{(0)}P^T)_w} \\
			&=&\sum_{t=0}^{T-1}\sum_{v\in V}\sum_{k=0}^{\conf^{(t)}_v-1}\sum_{u\in \trN_v}
			\left(\df<\rvD^{(t)}_v(k)=u>-\E\left[\df<\rvD^{(t)}_v(k)=u>\mid \conf^{(t)}_v\right]\right)\left(P^{T-t-1}_{u,w}-P^{T-t-1}_{v,w}\right)
			\end{eqnarray*}
			holds for any $\conf^{(0)}\in \mathbb{Z}^\Nov_{\geq 0}$, $T\geq 0$, $\conf^{(T)}$ of Algorithm~2 and $w\in V$.
			\label{lemm:discframe}
		\end{lemma}
		\begin{proof}
		%The main idea is combining the equations $\sum_{t=0}^{\tT-1}\left(\conf^{(t+1)}-\conf^{(t)}\trM\right)\trM^{T-t-1}=\conf^{(\tT)}-\conf^{(0)}\trM^\tT$ (telescoping),
		%$\conf^{(t+1)}_u=\sum_{v\in V}\sum_{k=0}^{\conf^{(t)}_v-1}\df<\rvD^{(t)}_v(k)=u>$ (Observation~\ref{obs:defconf}) 
		%and $(\conf^{(t)}P)_u=\sum_{v\in V}\sum_{k=0}^{\conf^{(t)}_v-1}\E[\df<\rvD^{(t)}_v(k)=u>\mid \conf^{(t)}_v]$ (Observation~\ref{obs:exD}). 
		%We put the proof of Lemma~\ref{lemm:discframe} in the full-version of this paper because of the page limitation.
%\if0
			Since
			\begin{eqnarray*}
				\sum_{t=0}^{\tT-1}\left(\conf^{(t+1)}-\conf^{(t)}\trM\right)\trM^{T-t-1} 
				&=& \sum_{t=0}^{\tT-1}\left(\conf^{(t+1)}\trM^{T-t-1}-\conf^{(t)}\trM^{T-t}\right) \nonumber \\
				&=& \conf^{(\tT)}\trM^0-\conf^{(0)}\trM^\tT
				\ =\ \conf^{(\tT)}-\conf^{(0)}\trM^\tT, 
				\label{eq:frame1}
			\end{eqnarray*}
			\begin{eqnarray}
				\conf^{(\tT)}_w-(\conf^{(0)}\trM^\tT)_w
				&=&\sum_{t=0}^{\tT-1}\sum_{u\in V}\left(\conf^{(t+1)}_u-(\conf^{(t)}\trM)_u\right)\trM^{T-t-1}_{u,w}
				\label{eq:lframe0}
			\end{eqnarray}
			holds. Then, from Observation~\ref{obs:defconf}, 
			\begin{eqnarray}
				\conf^{(t+1)}_u = \sum_{v\in V}\sum_{k=0}^{\conf^{(t)}_v-1}\df<\rvD^{(t)}_v(k)=u>
				\label{eq:lframe01}
			\end{eqnarray} holds and from Observation~\ref{obs:exD},
			\begin{eqnarray}
				(\conf^{(t)}P)_u &= &\sum_{v\in V}\conf^{(t)}_v\trM_{v,u} %\nonumber \\
				\ =\ \sum_{v\in V}\conf^{(t)}_v\sum_{k=0}^{\conf^{(t)}_v-1}\left|\left[\frac{k}{\conf^{(t)}_v}, \frac{k+1}{\conf^{(t)}_v}\right)\cap \IP_{v,u}\right| \nonumber \\
				&=&\sum_{v\in V}\sum_{k=0}^{\conf^{(t)}_v-1}\E\left[\df<\rvD^{(t)}_v(k)=u>\mid \conf^{(t)}_v\right].\label{eq:lframe00}
			\end{eqnarray}
			holds. Thus, from \eqref{eq:lframe0}--\eqref{eq:lframe00}, 
			\begin{eqnarray}
			\conf^{(\tT)}_w-(\conf^{(0)}\trM^\tT)_w
			&=&\sum_{t=0}^{\tT-1}\sum_{u\in V}\sum_{v\in V}\sum_{k=0}^{\conf^{(t)}_v-1}\left(\df<\rvD^{(t)}_v(k)=u>-\E\left[\df<\rvD^{(t)}_v(k)=u>\mid \conf^{(t)}_v\right]\right)\trM^{T-t-1}_{u,w}\nonumber \\
			&=&\sum_{t=0}^{\tT-1}\sum_{v\in V}\sum_{k=0}^{\conf^{(t)}_v-1}\sum_{u\in \trN_v}\left(\df<\rvD^{(t)}_v(k)=u>-\E\left[\df<\rvD^{(t)}_v(k)=u>\mid \conf^{(t)}_v\right]\right)\trM^{T-t-1}_{u,w}. \nonumber \\
			\label{eq:lframe22}
			\end{eqnarray}
			Then, since %Finally, it holds that
			\begin{eqnarray}
				\sum_{u\in \trN_v}\left(\df<\rvD^{(t)}_v(k)=u>-\E\left[\df<\rvD^{(t)}_v(k)=u>\mid \conf^{(t)}_v\right]\right)P^{T-t-1}_{v,w}&=&0,
				\label{eq:lframe3}
			\end{eqnarray}
			we obtain the claim by subtracting \eqref{eq:lframe3} from \eqref{eq:lframe22}.
			Note that we have
			$
				\sum_{u\in \trN_v} \df<\rvD^{(t)}_v(k)=u>=1\label{eq:rvD1}
			$
			and
			$
				\sum_{u\in \trN_v} \E[\df<\rvD^{(t)}_v(k)=u>\mid \conf^{(t)}_v] 
				=1 %\label{eq:ErvD1}
			$.
			
%\fi
		\end{proof}
		Next, we introduce the following three lemmas which are concerned with the conditional expectation of $\df<\rvDs_l=u>$ and $\E[\df<\rvDs_l=u>\mid \conf^{(t)}_v]$.
		Each of them is derived from our key Observations~\ref{obs:defconf}--\ref{obs:indepD}.
		%Observation~\ref{obs:indepD} (independency of destination operators) and Proposition~\ref{obs:exD} (expected value of destination operators), and ...def of conf.
		\begin{lemma}%[Key Lemma for MD, using independence]
			For any $t\in \{0,1,\ldots, T-1\}$, $v\in \{0,1,\ldots, \Nov - 1\}$ and $k\in \{0,1,\ldots, \Not-1\}$, 
			let $l=\Not \Nov t+\Not v+k$.
			Then, it holds that
			\begin{eqnarray*}
				\E\left[\df<\rvDs_l=u>\mid \rvDs_0, \rvDs_1, \ldots, \rvDs_{l-1}\right]&=&\E\left[\df<\rvDs_l=u>\mid \conf^{(t)}_v\right].
			\end{eqnarray*}
			\label{lemm:keylm1}
		\end{lemma}
		\begin{proof}
			From the chain rule of conditional expectations, we have
			\begin{eqnarray*}
				\E\left[\df<\rvDs_l=u>\mid \rvDs_0, \ldots, \rvDs_{l-1}\right]
				&=&\E\left[\df<\rvDs_l=u>\mid \rvDs_0, \ldots, \rvDs_{\Not \Nov t-1}, \rvDs_{\Not \Nov t},\ldots,  \rvDs_{l-1}\right]\\
				&=&\E\left[\df<\rvDs_l=u>\mid \rvDs_0, \ldots, \rvDs_{\Not \Nov t-1}\right]
				=\E\left[\df<\rvDs_l=u>\mid \conf^{(t)}_v\right].
			\end{eqnarray*}
			The second equality holds from the independency (Observation~\ref{obs:indepD}).
			The last equality holds from Observation~\ref{obs:confM}, i.e. $\rvDs_0, \ldots, \rvDs_{\Not \Nov t-1}$ determines $\conf^{(t)}_v$.
			%Thus we obtain the claim from 
		\end{proof}
		
		\begin{lemma}%[Key Lemma for MD, using independence]
			For any $t\in \{0,1,\ldots, T-1\}$, $v\in \{0,1,\ldots, \Nov - 1\}$ and $k\in \{0,1,\ldots, \Not-1\}$, 
			let $l=\Not \Nov t+\Not v+k$.
			Then, for any $\ell\in \{0,1,\ldots, \Not \Nov T-1\}$, it holds that
			\begin{eqnarray*}
				\lefteqn{\E\left[\df<\rvDs_l=u>\mid  \rvDs_0, \rvDs_1, \ldots, \rvDs_\ell \right]}\\
				&=&
				\begin{cases}
					\df<\rvDs_l=u> & ({\rm if}\ l\leq \ell)\\
					\E\left[\E\left[\df<\rvDs_l=u>\mid \conf^{(t)}_v\right]\mid \rvDs_0, \rvDs_1, \ldots, \rvDs_{\ell} \right] & ({\rm if}\ l \geq \ell+1)
				\end{cases}
			\end{eqnarray*}
			\label{lemm:keydiffdest}
		\end{lemma}
		\begin{proof}
			For the case $l\leq \ell$, the claim is true obviously.
			If $l\geq \ell +1$, it holds that $l-1\geq \ell$ and from the chain rule of the conditional expectation, we have
			\begin{eqnarray*}
				\E\left[\df<\rvDs_l=u>\mid  \rvDs_0, \ldots, \rvDs_\ell \right]
				&=&\E\left[\E\left[\df<\rvDs_l=u>\mid  \rvDs_0, \ldots, \rvDs_{l-1}\right]\mid  \rvDs_0, \ldots, \rvDs_\ell \right].
			\end{eqnarray*}
			Thus we obtain the claim from Lemma~\ref{lemm:keylm1}.%~\ref{obs:confM}, i.e. $\rvDs_0, \ldots, \rvDs_{\Not \Nov t-1}$ determines $\conf^{(t)}_v$.
		\end{proof}
		\begin{lemma}%[Lemma for \Not D]
			For any $t\in \{0,1,\ldots, T-1\}$, $v\in \{0,1,\ldots, \Nov - 1\}$ and $k\in \{0,1,\ldots, \Not-1\}$, let $l=\Not \Nov t+\Not v+k$. Then, for any $\ell\geq \Not \Nov t-1$, it holds that
			\begin{eqnarray*}
				\E\left[\E\left[\df<\rvDs_l=u>\mid \conf^{(t)}_v\right]\mid \rvDs_0, \rvDs_1, \ldots, \rvDs_{\ell} \right]
				&=&\E\left[\df<\rvDs_l=u>\mid \conf^{(t)}_v\right].
			\end{eqnarray*}
			%holds for any $l\in \{0,1,\ldots, \Not \Nov T-1\}$ and 
		\label{lemm:keydiffconf}
		\end{lemma}
		\begin{proof}
			Since $\E[\df<\rvDs_l=u>\mid \conf^{(t)}_v]$ is a function of $\conf^{(t)}_v$ (Observation~\ref{obs:exD}) 
			and $\conf^{(t)}_v$ is determined by $\rvDs_{0},\ldots, \rvDs_{\Not \Nov t-1}$ (Observation~\ref{obs:confM}), we obtain the claim.
		\end{proof}
		%%%%%%%%%%%%%%%%%%%%%%%%
		\begin{proof}[Proof of Lemma~\ref{lemm:mainM}]
			For any $t\in \{0,1,\ldots, T-1\}$, $v\in \{0,1,\ldots, \Nov - 1\}$ and $k\in \{0,1,\ldots, \Not -1\}$, let $l=\Not \Nov t+\Not v+k$.
			Then, for any $u\in \trN_v$ and $\ell \in \{0,1,\ldots, \Not \Nov T-1\}$, let
			\begin{eqnarray}
			\difm_u(l,\ell)
			&\defeq&\left(\E\left[\df<\rvDs_l=u>\mid \rvDs_0, \ldots, \rvDs_\ell\right]-\E\left[\E\left[\df<\rvDs_l=u>\mid \conf^{(t)}_v\right]\mid \rvDs_0, \ldots, \rvDs_\ell\right]\right)\nonumber \\
			&&\ -\ \left(\E\left[\df<\rvDs_l=u>\mid \rvDs_0, \ldots, \rvDs_{\ell-1}\right]-\E\left[\E\left[\df<\rvDs_l=u>\mid \conf^{(t)}_v\right]\mid \rvDs_0, \ldots, \rvDs_{\ell-1}\right]\right). \label{def:diffell}
			\end{eqnarray}
			From Lemma~\ref{lemm:discframe} and the definition of $\rvDs_\ell$ (\ref{def:rvDs}),  it holds that
			\begin{eqnarray}
				\lefteqn{ \E\left[\conf^{(T)}_w-(\conf^{(0)}P^T)_w\mid \rvDs_0, \ldots, \rvDs_\ell \right]
				-\E\left[\conf^{(T)}_w-(\conf^{(0)}P^T)_w\mid \rvDs_0, \ldots, \rvDs_{\ell-1} \right]}\nonumber \\
				&=& \sum_{t=0}^{T-1}\sum_{v=0}^{\Nov - 1}\sum_{k=0}^{\Not-1}\sum_{u\in \trN_v} \difm_u(\Not \Nov t+\Not v+k,\ell) \left(\trM^{\tT-t-1}_{u,w}-\trM^{\tT-t-1}_{v,w}\right).
				\label{eq:sumofkome}
			\end{eqnarray}
			To obtain the claim, we show that $\difm_u(\Not \Nov t+\Not v+k,\ell)=\difm_u(l,\ell)=0$ for any $l\neq \ell$ firstly. % and
			%$\difm_u(l,\ell)=\df<\rvDs_\ell=u>-\E\left[\df<\rvDs_\ell=u>\mid \conf^{(t)}_v\right]$ for $l=\ell$. 
			We consider this by the following \textbf{case 1.} ($l\leq \ell-1$) and  \textbf{case 2.} ($l\geq \ell+1$). % and \textbf{case 3.} ($l=\ell$).
			%, and $\difm_u(\ell,\ell)=$.
	
			\noindent \textbf{case 1. $l\leq \ell-1$:}
			In this case, it holds that both $l\leq \ell$ and $l\leq \ell-1$.
			Thus from Lemma~\ref{lemm:keydiffdest}, 
			\begin{eqnarray}
			\E\left[\df<\rvDs_l=u>\mid \rvDs_0, \ldots, \rvDs_\ell\right]&=&\df<\rvDs_l=u>\ \ \ {\rm and } \label{eq:exdd1}\\
			\E\left[\df<\rvDs_l=u>\mid \rvDs_0, \ldots, \rvDs_{\ell-1}\right]&=&\df<\rvDs_l=u> 
			\end{eqnarray}
			hold.
			Note that it holds that both $\ell\geq \Not \Nov t-1$ and $\ell-1\geq \Not \Nov t-1$ since $\ell -1 \geq l = \Not \Nov t+\Not v+k > \Not \Nov t-1$, from Lemma~\ref{lemm:keydiffconf}, we have
			\begin{eqnarray}
			\E\left[\E\left[\df<\rvDs_l=u>\mid \conf^{(t)}_v\right]\mid \rvDs_0, \ldots, \rvDs_\ell\right]&=&\E\left[\df<\rvDs_l=u>\mid \conf^{(t)}_v\right]\ \ \ {\rm and }\\
			\E\left[\E\left[\df<\rvDs_l=u>\mid \conf^{(t)}_v\right]\mid \rvDs_0, \ldots, \rvDs_{\ell-1}\right]&=&\E\left[\df<\rvDs_l=u>\mid \conf^{(t)}_v\right]. \label{eq:exdd2}
			\end{eqnarray}
			Thus $\difm_u(l,\ell)=0$ from \eqref{eq:exdd1}--\eqref{eq:exdd2} and \eqref{def:diffell} in this case. %\textbf{case 1. $l\leq \ell-1$}.
			
			\noindent \textbf{case 2. $l\geq \ell+1$:}		
			In this case, it holds that both $l\geq \ell+1$ and $l\geq (\ell-1)+1$.
			Using Lemma~\ref{lemm:keydiffdest}, we have
			\begin{eqnarray}
			\E\left[\df<\rvDs_l=u>\mid \rvDs_0, \ldots, \rvDs_\ell\right]&=&\E\left[\E\left[\df<\rvDs_l=u>\mid \conf^{(t)}_v\right]\mid \rvDs_0, \ldots, \rvDs_{\ell} \right]\ \ \ {\rm and }\label{eq:exdd3} \\
			\E\left[\df<\rvDs_l=u>\mid \rvDs_0, \ldots, \rvDs_{\ell-1}\right]&=&\E\left[\E\left[\df<\rvDs_l=u>\mid \conf^{(t)}_v\right]\mid \rvDs_0, \ldots, \rvDs_{\ell-1} \right].\label{eq:exdd4}
			\end{eqnarray}
			Thus $\difm_u(l,\ell)=0$  from \eqref{eq:exdd3}, \eqref{eq:exdd4} and \eqref{def:diffell}  in this case. %\textbf{case 2. $l\geq \ell+1$}.
			
			From these discussion, all term such that $l\neq \ell$ in \eqref{eq:sumofkome} is $0$, hence we have
			\begin{eqnarray}
				\eqref{eq:sumofkome}&=&\sum_{u\in \trN_\fv} \difm_u(\ell,\ell) \left(\trM^{\tT-\ft-1}_{u,w}-\trM^{\tT-\ft-1}_{\fv,w}\right).
			\end{eqnarray}
			Note that $\ell=\Not \Nov \ft +\Not \fv +\fk$ from the assumption of the Lemma~\ref{lemm:mainM}.
			We conclude the proof by showing that $\difm_u(\ell,\ell)=\df<\rvDs_\ell=u>-\E[\df<\rvDs_\ell=u>\mid \conf^{(\ft)}_\fv]$. 
			%
			%\noindent \textbf{case 3. $l=\ell$:}
			%In this case, both $l\leq \ell$ and $l\geq (\ell-1)+1$ hold. Thus
			We have
			\begin{eqnarray}
			\E\left[\df<\rvDs_\ell=u>\mid \rvDs_0, \ldots, \rvDs_\ell\right]&=&\df<\rvDs_\ell=u>\ \ \ {\rm and}\label{eq:exdd5} \\
			\E\left[\df<\rvDs_\ell=u>\mid \rvDs_0, \ldots, \rvDs_{\ell-1}\right]&=&\E\left[\df<\rvDs_\ell=u>\mid \conf^{(\ft)}_\fv\right],
			\end{eqnarray}
			where the second equality holds from Lemma~\ref{lemm:keylm1}.
			Note that both $\ell\geq \Not \Nov \ft-1$ and $\ell-1\geq \Not \Nov \ft-1$ hold since $\ell -1  = \Not \Nov \ft+\Not \fv+\fk-1 \geq \Not \Nov \ft-1$. 
			From these facts and Lemma~\ref{lemm:keydiffconf},  we have
			\begin{eqnarray}
			\E\left[\E\left[\df<\rvDs_\ell=u>\mid \conf^{(\ft)}_\fv\right]\mid \rvDs_0, \ldots, \rvDs_\ell\right]&=&\E\left[\df<\rvDs_\ell=u>\mid \conf^{(\ft)}_\fv\right]\ \ \ {\rm and}\\
			\E\left[\E\left[\df<\rvDs_\ell=u>\mid \conf^{(\ft)}_\fv\right]\mid \rvDs_0, \ldots, \rvDs_{\ell-1}\right]&=&\E\left[\df<\rvDs_\ell=u>\mid \conf^{(\ft)}_\fv\right].\label{eq:exdd6} 
			\end{eqnarray}
			Thus $\difm_u(\ell,\ell)=\df<\rvDs_\ell=u>-\E[\df<\rvDs_\ell=u>\mid \conf^{(\ft)}_\fv]$ from \eqref{eq:exdd5}--\eqref{eq:exdd6} and \eqref{def:diffell}, and we obtain the claim.
		\end{proof}

%\newpage
%
\section{General upper bound of the local 2-divergence}\label{sec:local2}
%
	%To get Theorems~\ref{thm:reg} and~\ref{thm:gen}, we need to give an upper bound of the local 2-divergence $\Psi_2(P)$.
	This section gives a general upper bound of $\Psi_2(P)$ for any irreducible, {\em reversible} and {\em lazy} chain.
	Let $\pi\in \mathbb{R}_{\geq 0}$ denotes the {\em stationary distribution} of the round matrix $P$, i.e. the probability distribution such that $\pi P=\pi$ holds.
	A $P$ is called {\em reversible} if $\pi_vP_{v,u}=\pi_uP_{u,v}$ holds for any $v,u\in V$. %This is a kind of generalized assumption of the symmetry.
	A $P$ is called {\em lazy} if $\trM_{v,v}\geq 1/2$ holds for any $v\in V$.

	For a reversible $P$, its stationary distribution $\pi$ and a vector $f\in \mathbb{R}^\Nov$, 
	the {\em Dirichlet form} ${\cal E}(f)$ is defined as follows.
	\begin{eqnarray}
	\label{def:dirichlet}
	{\cal E}(f)&\defeq &\frac{1}{2}\sum_{(v,u)\in V\times V}(f_v-f_u)^2\pi_vP_{v,u}.
	\end{eqnarray}
	Dirichlet form is concerned with many techniques to estimating the key measures of Markov chains such as eigenvalues (cf. \cite{AF02, LPW08}).
	Now, we consider the connection between the Dirichlet form and the local 2-divergence.
	Since there exists a $w\in V$ such that $\bigl(\Psi_2(P)\bigr)^2=\sum_{t=0}^\infty \sum_{(v,u)\in V\times V:P_{v,u}>0}(P^t_{v,w}-P^t_{u,w})^2$ from Definition~\ref{def:local2}, 
	\begin{eqnarray}
	\bigl(\Psi_2(P)\bigr)^2
	%&=&\sum_{t=0}^\infty \sum_{(v,u)\in V\times V:P_{v,u}>0}(P^t_{v,w}-P^t_{u,w})^2\nonumber \\
	&\leq &\sum_{t=0}^\infty \sum_{(v,u)\in V\times V}(P^t_{v,w}-P^t_{u,w})^2\frac{\pi_vP_{v,u}}{\min_{(v,u)\in V\times V:P_{v,u}>0}\pi_vP_{v,u}}\nonumber \\
	&=&\frac{2}{\min_{(v,u)\in V\times V:P_{v,u}>0}\pi_vP_{v,u}} \sum_{t=0}^\infty {\cal E}(P^t_{\cdot, w}).\label{eq:dirichlet}
	\end{eqnarray}
	holds by taking $f_x=P^t_{x,w}$.
	Now, we introduce the following main lemma to give an upper bound of the local 2-divergence.
	\begin{lemma}[Dirichlet form of $P^t_{\cdotp,w}$]
		\label{lemm:dirichlet}
		For any reversible $P$, $w\in V$ and $t\geq 0$, it holds that
		\begin{eqnarray*}
			%\sum_{v\in V}\sum_{u\in \trN_v}\pi_v\trM_{v,u}(\trM^t_{u,w}-\trM^t_{v,w})^2
			 {\cal E}(P^t_{\cdot, w})&=&\pi_w(\trM^{2t}_{w,w}-\trM^{2t+1}_{w,w}).
		\end{eqnarray*}
	\end{lemma}
	\begin{proof}
			From the definition of the Dirichlet form, it holds that
			\begin{eqnarray}
				2{\cal E}(P^t_{\cdot, w}) 
				&=&\sum_{v\in V}\sum_{u\in V}\pi_v\trM_{v,u}(\trM^t_{v,w}-\trM^t_{u,w})^2 \nonumber \\
				%&=& \sum_{v\in V}\sum_{u\in V}\pi_v\trM_{v,u}\left((\trM^t_{u,w})^2+(\trM^t_{v,w})^2-2\trM^t_{u,w}\trM^t_{v,w}\right) \nonumber\\
				&=& \sum_{v\in V}\sum_{u\in V}\pi_v\trM_{v,u}(\trM^t_{v,w})^2 + \sum_{v\in V}\sum_{u\in V}\pi_v\trM_{v,u}(\trM^t_{u,w})^2 %\nonumber\\
				-2 \sum_{v\in V}\sum_{u\in V}\pi_v\trM_{v,u}\trM^t_{v,w}\trM^t_{u,w}.\label{eq:local2eq1} %\nonumber \\
				%&=&\sum_{v\in V}\pi_v(\trM^t_{v,w})^2+\sum_{u\in V}\pi_u(\trM^t_{u,w})^2-2 \sum_{v\in V}\sum_{u\in V}\pi_w\trM_{v,u}\trM^t_{w,v}\trM^t_{u,w} \nonumber \\
				%&=&2\sum_{v\in V}\pi_v\trM^t_{v,w}\trM^t_{v,w}-2 \pi_wP^{2t+1}_{w,w}
			\end{eqnarray}
			Then, from the reversibility of $P$,
			\begin{eqnarray}
				\lefteqn{\sum_{v\in V}\sum_{u\in V}\pi_v\trM_{v,u}(\trM^t_{v,w})^2+\sum_{v\in V}\sum_{u\in V}\pi_v\trM_{v,u}(\trM^t_{u,w})^2  }\nonumber \\
				&=&\sum_{v\in V}\sum_{u\in V}\pi_v\trM_{v,u}(\trM^t_{v,w})^2+\sum_{v\in V}\sum_{u\in V}\pi_u\trM_{u,v}(\trM^t_{u,w})^2%\nonumber \\
				\ =\ \sum_{v\in V}\pi_v(\trM^t_{v,w})^2+\sum_{u\in V}\pi_u(\trM^t_{u,w})^2 \nonumber \\
				&=&\sum_{v\in V}\pi_w\trM^t_{w,v}\trM^t_{v,w}+\sum_{u\in V}\pi_w\trM^t_{w,u}\trM^t_{u,w} %\nonumber \\
				\ =\ 2\pi_w\trM^{2t}_{w,w}\ \ \ {\rm and} 
				\label{eq:local2eq2} 
			\end{eqnarray}
			and
			\begin{eqnarray}
				2\sum_{v\in V}\sum_{u\in V}\pi_v\trM_{v,u}\trM^t_{v,w}\trM^t_{u,w}
				&=&2\pi_w\sum_{v\in V}\sum_{u\in V}\trM^t_{w,v}\trM_{v,u}\trM^t_{u,w} %\nonumber \\
				\ =\ 2\pi_wP^{2t+1}_{w,w}\label{eq:local2eq3}
			\end{eqnarray}
			hold. Combining \eqref{eq:local2eq1}-\eqref{eq:local2eq3}, we obtain the claim.
		\end{proof}
		%Combining $P^{t}_{w,w}\geq P^{t+1}_{w,w}$ holds for lazy $P$. 
		\begin{proof}[Proof of Theorem~\ref{thm:local2}]
			%From Lemma~\ref{lemm:dirichlet}, we have
			Combining \eqref{eq:dirichlet}, Lemma~\ref{lemm:dirichlet} and the assumption of laziness of $P$, 
			\begin{eqnarray*}
				\sum_{t=0}^\infty {\cal E}(P^t_{\cdot, w})
				&=&\pi_w \sum_{t=0}^{\infty} (P^{2t}_{w,w}-P^{2t+1}_{w,w}) %\nonumber \\
				\ \leq\ \pi_w \sum_{t=0}^{\infty} (P^{2t}_{w,w}-P^{2t+2}_{w,w}) 
				\ \leq\ \pi_w
				\label{eq:eq11}
			\end{eqnarray*}
			holds, thus we obtain the Theorem~\ref{thm:local2}.
			Note that $P^{t}_{w,w}\geq P^{t+1}_{w,w}$ holds for any lazy $P$ (cf. Proposition 10.25 of \cite{LPW08}).
			%Thus we obtain the claim from (\ref{eq:eq11}).
		\end{proof}
\section{Upper bound on the discrepancy for symmetric round matrices}\label{sec:final}
	Now, we conclude this paper stating the proofs of Theorems~\ref{thm:reg} and~\ref{thm:gen}.
	First, we show the following general theorem according to the discrepancy for any symmetric round matrix.
	\begin{theorem}[Result for symmetric matrices]
		\label{thm:maindisc}
		Suppose that $P$ is irreducible and symmetric.
		Then, for any $\conf^{(0)}$ and for each $T\geq \frac{\log(4\Disc(\conf^{(0)}) \Nov)}{1-\sE}$, 
		$\conf^{(T)}$ of Algorithm~2 satisfies that 
		\begin{eqnarray*}
			\Pro\left[ \Disc(\conf^{(T)}) \leq 9\Psi_2(\trM)\sqrt{\log \Nov} \right] \geq 1-\frac{2}{\Nov}.
		\end{eqnarray*}
	\end{theorem}
	\begin{proof}
		Since $\conf^{(T)}_v-\conf^{(T)}_u=(\conf^{(T)}_v-(\conf^{(0)}P^T)_v)+((\conf^{(0)}P^T)_v-(\conf^{(0)}P^T)_u)+((\conf^{(0)}P^T)_u-\conf^{(T)}_u)$, we have
		\begin{eqnarray*}
			\left|\conf^{(T)}_v-\conf^{(T)}_u\right|&\leq &
			\left|\conf^{(T)}_v-(\conf^{(0)}P^T)_v\right|+\left|(\conf^{(0)}P^T)_v-(\conf^{(0)}P^T)_u\right|+\left|(\conf^{(0)}P^T)_u)-\conf^{(T)}_u\right|\\
			&\leq &2\Disc(\conf^{(T)})+\left|(\conf^{(0)}P^T)_v-(\conf^{(0)}P^T)_u\right|
		\end{eqnarray*}
		Combining Proposition~\ref{prop:cont} and Theorem~\ref{thm:main}, we obtain the claim.
	\end{proof}
	%Combining Theorem~\ref{thm:maindisc} and Corollary~\ref{cor:local2}, we obtain the proof of Theorems~\ref{thm:reg} and~\ref{thm:gen}.
	\begin{proof}[Proof of Theorem~\ref{thm:reg}]
		Since Algorithm~1 is Algorithm~2 according to the transition matrix of the lazy random walk on $G=(V,E)$, i.e.
		$P$ such that $P_{v,u}=1/(2d)$ for any $\{v,u\}\in E$, $P_{v,v}=1/2$ for any $v\in V$, and $P_{v,u}=0$ for any $\{v,u\}\notin E$.
		For this $P$, since $\min_{(v,u)\in V\times V:P_{v,u}>0}P_{v,u}=1/(2d)$, 
		$\Psi_2(P)\leq 2\sqrt{d}$ holds from Corollary~\ref{cor:local2}. Thus we obtain the claim by Theorem~\ref{thm:maindisc}.
	\end{proof}
	\begin{proof}[Proof of Theorem~\ref{thm:gen}]
		From the definition of $\trMM$ in Section~\ref{sec:genresult}, this chain is lazy and symmetric.
		Since $\min_{(v,u)\in V\times V:P_{v,u}>0}P_{v,u}=1/(2d_{\max})$,
		$\Psi_2(P)\leq 2\sqrt{d_{\max}}$ holds from Corollary~\ref{cor:local2}. Thus we obtain the claim by Theorem~\ref{thm:maindisc}.
	\end{proof}
%\section{Concluding and remarks}
%%%%%%%%%%%%%%%%%%%%%%%%%%
\section*{Acknowledgements}%
%%%%%%%%%%%%%%%%%%%%%%%%%%
This work is supported by JSPS KAKENHI Grant Number 17H07116.
\bibliographystyle{abbrv}

\appendix
\section{Proof of Proposition~\ref{prop:cont}}\label{sec:continuous2}
		\if0\begin{proposition}[The discrepancy of the continuous diffusion algorithm \cite{SS94, RSW98}]
		\label{prop:cont}
			Suppose that $P$ is irreducible and symmetric.
			Then, for any $\cconf^{(0)}\in \mathbb{R}^n_{\geq 0}$, $w\in V$, $\varepsilon\in (0,1)$ and 
			$
				T \geq \frac{1}{1-\sE}\log\left(\frac{4\Disc(\cconf^{(0)}) \Nov}{\varepsilon}\right), 
			$
			$
				\Disc(\cconf^{(T)})= \Disc(\cconf^{(0)} P^T)\leq \varepsilon
			$ holds.
		\end{proposition}\fi
		%We check the following proposition in the full version.
%NOTENOTENOTENOTENOTENOTENOTENOTENOTENOTENOTENOTENOTENOTENOTENOTENOTENOTE
			\begin{proof}
			We have
			\begin{eqnarray*}
				\left| (\cconf^{(0)}P^T)_w -\frac{\Not}{\Nov} \right|
				&=&\left| \sum_{v\in V}\cconf^{(0)}_v\left(\trM^T_{v,w}-\frac{1}{\Nov}\right) \right| %\\
				\ =\ \left| \sum_{v\in V}\cconf^{(0)}_v\left(\trM^T_{w,v}-\frac{1}{\Nov}\right) \right| \\
				&=&\left| \sum_{v\in V}(\cconf^{(0)}_v-\cconf^{(0)}_x)\left(\trM^T_{w,v}-\frac{1}{\Nov}\right) \right| %\\
				\ \leq \ 2\Disc(\cconf^{(0)}) \cdotp \frac{1}{2}\sum_{v\in V}\left|\trM^{T}_{w,v}-\frac{1}{\Nov}\right|.
			\end{eqnarray*}	
			Note that we used the assumption of the symmetry of $\trM$ and $\sum_{v\in V}\cconf^{(0)}_x\left(\trM^T_{w,v}-\frac{1}{\Nov}\right)=0$ for any $x\in V$.
			Using Theorem 12.4 in \cite{LPW08},  
			$\frac{1}{2}\sum_{v\in V}\left|\trM^{T}_{w,v}-\frac{1}{\Nov}\right|\leq \varepsilon'$
			holds for any $T\geq \frac{1}{1-\sE} \log \left(\frac{\Nov }{\varepsilon'}\right)$.
			%we have $\frac{1}{2}\sum_{v\in V}\left|\trM^{T}_{w,v}-\frac{1}{\Nov}\right|\leq \varepsilon$ if $T\geq \frac{1}{1-\sE}\log\left(\frac{\Nov }{\varepsilon}\right)}$.
			Thus, with the assumption of $T$, we have $\frac{1}{2}\sum_{v\in V}\left|\trM^{T}_{w,v}-\frac{1}{\Nov}\right|\leq \frac{\varepsilon}{4\Disc(\cconf^{(0)})}$.
			Hence we have 
			\begin{eqnarray*}
				\left| (\cconf^{(0)}P^T)_v -(\cconf^{(0)}P^T)_u \right|
				&=&\left| \left((\cconf^{(0)}P^T)_v-\frac{\Not}{\Nov}\right) - \left(\frac{\Not}{n}-(\cconf^{(0)}P^T)_u\right) \right| \\
				&\leq &\left| (\cconf^{(0)}P^T)_v -\frac{\Not}{\Nov} \right|+\left| (\cconf^{(0)}P^T)_u -\frac{\Not}{\Nov} \right|
				\ \leq \ \frac{\varepsilon}{2}+\frac{\varepsilon}{2} =\varepsilon
			\end{eqnarray*}	
			for any $v,u\in V$, and we obtain the claim.
			\end{proof}
\section{Concentration inequality}\label{sec:concentration}
		\begin{theorem}[Asuma-Hoeffding Inequality, \cite{MU17}]
			Let $X_0, \ldots, X_n$ be a martingale such that
			\begin{eqnarray*}
			|X_k-X_{k-1}|\leq c_k.
			\end{eqnarray*}
			Then, for all $t\geq 1$ and any $\lambda>0$, 
			\begin{eqnarray*}
				\Pro\left[|X_t-X_0|\geq \lambda \right]\leq 2\exp\left[-\frac{\lambda^2}{2\sum_{k=1}^t(c_k)^2}\right].
			\end{eqnarray*}
		\end{theorem}

\end{document}